\documentclass[11pt]{article}

\usepackage{authblk}
\usepackage{amsmath, amsthm, amssymb, amsfonts}
\usepackage[pdftex]{graphicx}
\usepackage[english]{babel}
\usepackage{color}
\usepackage{bbold}
\usepackage{array}
\usepackage{multirow}
\usepackage{pdfpages}
\usepackage[subnum]{cases}
\usepackage[left=2.5cm,right=2.5cm,top=3.0cm,bottom=3.0cm]{geometry}
\usepackage{tikz}
\usepackage{paralist}
\usepackage[colorinlistoftodos,bordercolor=orange,
            backgroundcolor=orange!20,linecolor=orange,
            textsize=normalsize]{todonotes}
            
\usepackage[T1]{fontenc}
\usepackage[utf8]{inputenc}
\usepackage{lmodern}
\usepackage{makecell}
\newcolumntype{L}[1]{>{\raggedright\let\newline\\\arraybackslash\hspace{0pt}}m{#1}}
\newcolumntype{C}[1]{>{\centering\let\newline\\\arraybackslash\hspace{0pt}}m{#1}}
\newcolumntype{R}[1]{>{\raggedleft\let\newline\\\arraybackslash\hspace{0pt}}m{#1}}
\usetikzlibrary{snakes,arrows,shapes}
\newtheorem{theorem}{Theorem}
\newtheorem{definition}{Definition}
\newtheorem{corollary}{Corollary}
\newtheorem{proposition}{Proposition}
\newtheorem{lemma}{Lemma}
\newtheorem{example}{Example}

\usepackage[varg]{txfonts}

\title{On Strong Equilibria and Improvement Dynamics \\ in Network Creation Games}
                            
\author[1]{Tomasz Janus}
\author[2]{Bart de Keijzer}
\affil[1]{Department of Computer Science, \authorcr
The University of Warwick, Coventry, United Kingdom \authorcr
\texttt{t.janus@warwick.ac.uk}}
\affil[2]{School of Computer Science and Electronic Engineering, \authorcr
	University of Essex, Colchester, United Kingdom, \authorcr \texttt{b.dekeijzer{@}essex.ac.uk}}

\begin{document}
\pagestyle{plain}
\bibliographystyle{plain}
\sloppy 
\maketitle

\begin{abstract}
We study strong equilibria in network creation games. These form a classical and well-studied class of games where a set of players form a network by buying edges to their neighbors at a cost of a fixed parameter $\alpha$. The cost of a player is defined to be the cost of the bought edges plus the sum of distances to all the players in the resulting graph. 

We identify and characterize various structural properties of strong equilibria, which lead to a characterization of the set of strong equilibria for all $\alpha$ in the range $(0,2)$. For $\alpha > 2$, Andelman et al. (2009) prove that a star graph in which every leaf buys one edge to the center node is a strong equilibrium, and conjecture that in fact \emph{any} star is a strong equilibrium. We resolve this conjecture in the affirmative. Additionally, we show that when $\alpha$ is large enough ($\geq 2n$) there exist non-star trees that are strong equilibria. For the strong price of anarchy, we provide precise expressions when $\alpha$ is in the range $(0,2)$, and we prove a lower bound of $3/2$ when $\alpha \geq 2$.

Lastly, we aim to characterize under which conditions (coalitional) improvement dynamics may converge to a strong equilibrium. To this end, we study the (coalitional) finite improvement property and (coalitional) weak acyclicity property. We prove various conditions under which these properties do and do not hold. Some of these results also hold for the class of pure Nash equilibria. 
\end{abstract}

\section{Introduction}
The Internet is a large-scale network that has emerged mostly from the spontaneous, distributed interaction of selfish agents. Understanding the process of creating of such networks is an interesting scientific problem. Insights into this process may help to understand and predict how networks emerge, change, and evolve. This holds in particular for social networks. 

The field of game theory has developed a large number of tools and models to analyze the interaction of many independent agents. The Internet and many other networks can be argued to have formed through interaction between many strategic agents. It is therefore natural to use game theory to study the process of network formation. Indeed, this has been the subject of study in many research papers, e.g. \cite{fabrikant,albers,balagoyal,alon,galeotti,derks,lenzner}, to mention only a few of them.

We focus here on the classical model of \cite{fabrikant}, which is probably the class of network formation game that is most prominently studied by algorithmic game theorists. This model stands out due to its simplicity and elegance: It is defined as a game on $n$ players, where each player may choose an arbitrary set of edges that connects herself to a subset of other players, so that a graph forms where the vertices are the players. Buying any edge costs a fixed amount $\alpha \in \mathbb{R}$, which is the same for every player. Now, the cost of a player is defined as the total cost of set of edges she bought, plus the sum of distances to all the other players in the graph. A network creation game is therefore determined by two parameters: $\alpha$ and $n$.

Another reason for why these network creation games are an interesting topic of study, comes from the surprisingly challenging questions that emerge from this simple class of games. For example, it is (as of writing) unknown whether the \emph{price of anarchy} of these network creation games is bounded by a constant, where price of anarchy is defined as the factor by which the total cost of a pure Nash equilibrium is away from the minimum possible total cost \cite{poa1,poa2}.

In the present work, we study \emph{strong equilibria}, which are a refinement of the pure Nash equilibrium solution concept. 
Strong equilibria are defined as pure Nash equilibria that are resilient against strategy changes that are made collectively by arbitrary \emph{sets} of players, in addition to strategy changes that are made by \emph{individual} players (see \cite{SEdef}). 
Generally, such an equilibrium may not exist, since this is already the case for pure Nash equilibria. On the other hand, in case they do exist, strong equilibria are extremely robust, and they are likely to describe the final outcome of a game provided that they are, in a realistic sense, ``easy to attain'' for the players. Fortunately, as \cite{andelman} points out, in network creation games, strong equilibria are guaranteed to exist except for a very limited number of cases. The combination of the facts that strong equilibria are robust, and are almost always guaranteed to exist, calls for a detailed study of these equilibria in network creation games, which is what we do in the present work.

We provide in this paper a complete characterization of the set of all strong equilibria for $\alpha \in (0,2)$. Moreover, for $\alpha > 2$ we prove in the affirmative the conjecture of \cite{andelman} that any strategy profile that forms a star graph (i.e., a tree of depth $1$) is a strong equilibrium. We also show that for large enough $\alpha$ (namely, for $\alpha \geq 2n$), there exist strong equilibria that result in trees that are not stars. 

The price of anarchy restricted to strong equilibria is called the \emph{strong price of anarchy}. This notion was introduced in \cite{andelman}, where also the strong price of anarchy of network creation games was studied first. The authors prove there that the strong price of anarchy is at most $2$. We contribute to the understanding of the strong price of anarchy by providing a sequence of examples of strong equilibria where the strong price of anarchy converges to $3/2$, thereby providing the first non-trivial lower bound (to the best of our knowledge).  

Regarding the reachability and the likelihood for the players to actually attain a strong equilibrium, we study the question whether they can be reached by \emph{response dynamics}, that is, the process where we start from any strategy profile, and we repeatedly let a player or a set of players make a change of strategies that is beneficial for each player in the set, i.e., decreases their cost. In particular, we are interested in whether network creation games posess the \emph{coalitional finite improvement property} (that is: whether such response dynamics are guaranteed to result in a strong equilibrium), and the \emph{coalitional weak acyclicity property} (that is: whether there exists a sequence of coalitional strategy changes that ends in a strong equilibrium when starting from any strategy profile). We prove various conditions under which these properties are satisfied. Roughly, we show that coalitional weak acyclicity holds when $\alpha \in (0,1]$ or when starting from a strategy profile that forms a tree (for $\alpha \in (0,n/2]$), but that the coalitional finite improvement property is unfortunately not satisfied for any $\alpha$. Some of these results hold for pure Nash equilibria as well.\footnote{As for the study of pure Nash equilibria in network creation games, we additionally we show that there exists a strong equilibrium that results in a well-known strongly regular graph called the \emph{Hoffman--Singleton graph} \cite{hoffman}. This fact is relevant as an important topic of research in network creation games concerns the existence and properties of non-tree equilibria. The Hoffman--Singleton graph is now the smallest known example of a non-tree graph that is formed by a strict Nash equilibrium. This result can be found in Appendix~\ref{hoffmansingleton}.}

\subsection{Our Contributions}
A key publication that is strongly related to our work is \cite{andelman}, where the authors study the existence of strong equilibria in network creation games. The authors prove that the strong price of anarchy of network creation games does not exceed $2$ and provide insights into the structure and existence of strong equilibria. This is to the best of our knowledge the only paper studying strong equilibria in this particular type of network creation games.\footnote{However, the strong equilibrium concept has been studied in other types of games that concern network creation; see the related literature section below.}
Let us therefore summarize how the present paper complements and contributes to the results in \cite{andelman}: First, we provide additional results on the strong equilibrium structure, such that together with the results from \cite{andelman} we obtain a characterization of strong equilibria for $\alpha \in (0,2)$. Furthermore, in \cite{andelman} it was conjectured that all strategy profiles that form a star (and such that no edge is bought by two players at the same time) are strong equilibria. We provide a proof of this fact. Because \cite{andelman} does not provide examples of strong equilibria that are not stars (for $\alpha > 2$), this may suggest the conjecture that \emph{all} strong equilibria form a star for $\alpha > 2$. We show however that the latter is not true: We provide a family of examples of strong equilibria which form trees of diameter four (hence, not stars). More interestingly, the latter sequence of examples has a price of anarchy that converges to $3/2$, thereby providing (again, to the best of our knowledge) the first non-trivial lower bound on the strong price of anarchy. 

A second theme of our paper is to investigate under which circumstances the coalitional finite improvement and coalitional weak acyclicity properties are satisfied. We show to this end that coalitional weak acyclicity always holds for $\alpha \in (0,1]$ and holds for $\alpha \in (1, n/2)$ in case the starting strategy profile is a tree. We prove on the negative side that for all $\alpha$ there exists a number of players $n$ such that the coalitional finite improvement property does not hold. The only special case for which we manage to establish existence of the coalitional finite improvement property is for $n = 3$ and $\alpha > 1$. With regard to convergence of response dynamics to strong equilibria, an interesting question that we leave open is whether the coalitional weak acyclicity property holds for $\alpha > n/2$, and for $\alpha \in (1,n/2)$ when starting at non-tree strategy profiles. Some of our results on these properties also hold for the set of pure Nash equilibria.

Our results are summarized in the tables~\ref{SEsum}--~\ref{GD}. Table~\ref{SEsum} provides an overview for our characterization and structure theorems for strong equilibria, Table~\ref{SPoA} shows our bounds on the strong price of anarchy, and Table~\ref{GD} shows our results on the finite improvement and weak acyclicity properties of network creation games.

\begin{table}
\centering
 \small
\begin{tabular}{|C{2.8cm}|C{2.8cm}|C{2.8cm}|C{2.8cm}|C{2.8cm}|}
\cline{2-5}
\multicolumn{1}{c|}{}&$\alpha \in (0,1) $& $\alpha = 1$ & $\alpha \in (1,2)$ & $\alpha \ge 2$\\ 
\hline 
strong equilibria&Characterized (in \cite{andelman}, see Lemma~\ref{slt1} in present paper) & Characterized (Theorem~\ref{se1}) & Characterized (Proposition~\ref{prop:se12}) & 
Every star is a strong equilibrium (Theorem~\ref{strongge2}), existence of non-star strong equilibria (Theorem~\ref{thm:selowerbound})\\
\hline
\end{tabular}
\vspace*{0.1cm}
\caption{Overview of strong equilibria characterization results and structural results.} \label{SEsum}
\end{table}


\begin{table}
\centering
 \small
\begin{tabular}{|C{2.8cm}|C{2.8cm}|C{2.8cm}|C{2.8cm}|C{2.8cm}|}
\cline{2-5}
\multicolumn{1}{c|}{}&$\alpha \in (0,1) $& $\alpha = 1$ & $\alpha \in (1,2)$ & $\alpha \ge 2$\\ 
\hline 
strong price of anarchy &\thead{1 (Trivial, \\ Proposition~\ref{spoal1})} & $10/9$ if $n \leq 4$ and $(3n+2)/3n$ if $n \geq 5$ (Theorem~\ref{spoa1}) & $(2\alpha + 8)/(3\alpha + 6)$ if $n = 3$, and $(4\alpha + 16)/(6\alpha + 12)$ if $n=4$ (Proposition~\ref{spoa12}) & At least $3/2$ (Theorem~\ref{thm:spoa32}) and at most 2 \cite{andelman} \\
\hline
\end{tabular}
\vspace*{0.1cm}
\caption{Overview of bounds on the strong price of anarchy.} \label{SPoA}
\end{table}


\begin{table}
\centering
 \small
\begin{tabular}{|C{2.3cm}|C{2.3cm}|C{2.3cm}|C{2.3cm}|C{2.3cm}|C{2.3cm}|}
\cline{2-6}
\multicolumn{1}{c|}{}&$\alpha \in (0,1) $&$\alpha = 1$&$\alpha \in (1,2)$&$\alpha = 2$&$\alpha > 2$\\ 
\hline
\multirow{2}{*}{c-FIP} & \multirow{2}{1.5cm}{Negative (Lemma~\ref{cfiplt1})} & \multirow{2}{1.5cm}{Negative (Lemma~\ref{prop-cFIP1})} & Negative (Lemma~\ref{cfip12}) & Negative (Lemma~\ref{prop-cFIP2}) &
Negative (in \cite{Brandes2008}, see Corollary~\ref{cor:cfip2} in present paper)\\
\cline{4-6}
&&&\multicolumn{3}{C{6cm}|}{Positive for $n=3$ (Lemma~\ref{509a})} \\
\hline 
c-weak acyclicity & Positive (Corollary~\ref{cwaclt1}) & Positive (Proposition~\ref{cwac1}) &
\multicolumn{3}{C{6cm}|}{Positive with respect to trees for $\alpha \in (1, n/2)$ (Lemma~\ref{wat})}\\
\hline
\end{tabular}
\vspace*{0.1cm}
\caption{Summary of results on the c-FIP and c-weak acyclicity of network creation games.} \label{GD}
\end{table}

The outline of this paper is as follows. In the next section, we discuss various additional research papers that are related to our work. In Section \ref{sec:prelims}, we provide the reader with the necessary technical background, where we state e.g. the formal definitions of network creation games, strong equilibria, price of anarchy, and some graph theory notions. Sections \ref{structure} to \ref{dynamics} contain our main contributions: our structural results (summarized in Table \ref{SEsum}) are proved in Section \ref{structure}; our bounds on the strong price of anarchy (Table \ref{SPoA}) are proved in Section \ref{spoa}; and our results on improvement dynamics (see Table \ref{GD}) are proved in Section \ref{dynamics}. Lastly, in Section \ref{sec:discussion}, we provide various interesting open questions for future investigation.

\section{Related Literature}\label{sec:related}
We already discussed the work~\cite{andelman}. Network creation games were first defined in~\cite{fabrikant}. Moreover, \cite{fabrikant} conjectured that there exists an $A \in \mathbb{R}_{\geq 0}$ such that all \emph{non-transient equilibria} are trees for $\alpha \geq A$. (Transient equilibria are equilibria for which there exists a sequence of player deviations that leads to a non-equilibrium, such that all of the deviations in the sequence do not increase the deviating players' cost.) 

This conjecture was subsequently disproved by~\cite{albers}, where the authors construct non-tree equilibria for abitrarily high $\alpha$. These equilbiria are \emph{strict} (i.e., for no player there is a deviation that keeps her cost unchanged) and therefore non-transient, and their construction uses finite affine planes. Moreover, the authors show that the price of anarchy is constant for $\alpha \leq \sqrt{n}$ and for $\alpha \geq 12n\log n$. In the latter case they prove that any pure equilibrium is a tree. This bound was improved in \cite{matusz}, where it was shown that for $\alpha \geq 273 n$ all pure equilibria are trees. Later on, in~\cite{matusz2}, this was further improved by showing that it even holds for $\alpha \geq 65n$. Recently, two preprints have been released that make further progress towards proving that the price of anarchy of network creation games is constant: First, in~\cite{messegue}, \`{A}lvarez \& Messegu\'{e} show that every pure Nash equilibrium is a tree already when $\alpha > 17n$, and that the price of anarchy is bounded by a constant for $\alpha > 9n$. Shortly after that, Bil\`{o} \& Lenzner \cite{lenznerultranew} proved that all Nash equilibria are trees for $\alpha > 4n - 13$ and additionally provide an improved upper bound on the price of anarchy of tree equilibria. In a very recent preprint \cite{alvarez2018constant}, the bound on $\alpha$ has been further improved, as the authors show that the price of anarchy is constant if $\alpha > n(1 + \epsilon)$, for all $\epsilon > 0$.

In \cite{demaine}, some constant bounds on the price of anarchy were improved, and it was shown that for $\alpha \leq n^{1-\epsilon}$ the price of anarchy is constant, for all $\epsilon \geq 0$.
It remains an open question whether the price of anarchy is constant for all $\alpha \in \mathbb{R}_{\geq 0}$. In particular, the best known bound on the price of anarchy for $\alpha \in [n^{1-\epsilon}, 4n-13]$ is $2^{O\left(\sqrt{\log n}\right)}$, shown in \cite{demaine}. For all other choices of $\alpha$ the price of anarchy is known to be constant. The master's thesis \cite{buisan} provides some simplified proofs for some of the above facts, and proves that if an equilibrium graph has bounded degree, then the price of anarchy is bounded by a constant. It also studies some related computational questions.

Many other variants of network creation games have been considered as well. A version where disconnected players incur a finite cost rather than an infinite one was studied in \cite{Brandes2008}. In \cite{albers}, a version is introduced where the distance cost of a player $i$ to another player $j$ is weighted by some number $w_{ij}$. A special case of this weighted model was proposed in \cite{matusz3}. The paper \cite{demaine} introduces a version of the game where the distance cost of a player is defined the \emph{maximum} distance from $i$ to any other player (instead of the sum of distances), and studies the price of anarchy for these games. Further results on those games can be found in \cite{matusz}.
Another natural variant of a cost sharing game is one where both endpoints of an edge can contribute to its creation, as proposed in \cite{matusz3}, or must share its creation cost equally as proposed in \cite{corboparkes} and further investigated in \cite{demaine}.
In \cite{balagoyal}, a version of the game is studied where the edges are directed, and the distance of a player $i$ to another player $j$ is the minimum length of a directed path from $i$ to $j$. The literature on these games and generalizations thereof (see e.g., \cite{derks,derks2,billand}) concerns existence of equilibria and the properties of response dynamics. See \cite{balagoyal2,galeotti,hallersarangi,halleretal} for other undirected network creation models and properties of pure equilibria in those models. Further, in the very recent paper \cite{lenzner4}, a variant of network creation games is studied where the cost of buying an edge to a player is proportional to the number of neighbors of that player.

In \cite{alon}, the authors analyze the outcomes of the game under the assumption that the players consider deviations by swapping adjacent edges. Better response dynamics under this assumption have been studied in \cite{lenzner}. A modified version of this model is introduced in \cite{matusz3}, where players can only swap their \emph{own} edges. The authors prove some structural results on the pure equilibria that can then arise. Furthermore, in \cite{lenzner2} the deviation space is enriched by allowing the players to \emph{add} edges, and various price of anarchy type bounds are established under this assumption. In \cite{lenzner3}, the dynamics of play in various versions of network creation games are further investigated. 

The strong equilibrium concept has been studied for certain variants of network creation games. For instance, the paper \cite{dutta} studies a variant where both endpoints must agree on creating an edge and the objective is a general function depending on the network structure. This objective is split among the players according to a given value function. Further, a stronger solution concept is studied for this type of games in \cite{jackson}. In \cite{Avrachenkov2016} the authors study a strongly related coalition based game where stability is considered up to a given family of possibly deviating coalitions, and the focus is on best response dynamics where iteratively a random coalition from this family may deviate. Lastly, \cite{Avrachenkov2016b} studies a variation where multiple sets of players may act non-selfishly and share the same valuation function.

\section{Preliminaries}\label{sec:prelims}
A \emph{network creation game} $\Gamma$ is a game played by $n \geq 3$ players where the strategy set $\mathcal{S}_i$ of a player $i \in [n] = \{1, \ldots, n\}$ is given by
$\mathcal{S}_i = \{s : s \subseteq [n]\setminus\{i\}\}$. That is, each player chooses a subset of other players. Let $\mathcal{S} = \times_{i \in [n]} \mathcal{S}_i$ be the set of strategy profiles of $\Gamma$ and for a subset $K \subseteq [n]$ of players let $\mathcal{S}_{K} = \times_{i \in K} \mathcal{S}_i$. 
Given a strategy profile $s \in \mathcal{S}$, we define $G(s)$ as the undirected graph with vertex set $[n]$ and edge set $\{\{i,j\} : j \in s_i \vee i \in s_j\}$. For a graph $G$ on vertex set $[n]$, we denote by $d_G(i,j)$ the length of the shortest path from $i$ to $j$ in $G$ (and we define $d_G(i,i) = 0$ and the distance between two disconnected vertices as infinity). 

The cost of player $i$ under $s$ is given by $c_i(s) = c_i^b(s_i)  + c_i^d(s)$, where $c_i^b(s_i) = \alpha |s_i|$ is referred to as the \emph{building cost}, $\alpha \in \mathbb{R}_{\geq 0}$ is a player-independent constant, and $c_i^d(s) = \sum_{j = 1}^n d_{G(s)}(i,j)$ is referred to as the \emph{distance cost}. The interpretation given to this game is that players create a network by buing edges to other players. Buying a single edge costs $\alpha$. The shortest distance $d_{G(s)}(i,j)$ to each other player $j$ is furthermore added to the cost of a player $i$. Note that the distance between two vertices in separate connected components of a graph is defined to be infinity. Therefore, a player will experience an infinitely high cost in case the graph that forms is not connected. We denote a network creation game by the pair $(n,\alpha)$. 

For a strategy profile $s \in \mathcal{S}$ let $d(s) = \sum_i c_i^d(s)$. The social cost of strategy profile $s$, denoted $C(s)$, is defined as the sum of all individual costs: $C(s) = \sum_{i \in [n]} c_i(s) = \alpha \sum_i |s_i| + d(s)$.

We study the \emph{strong equilibria} of this game, which are formally defined as follows.
\begin{definition}
A \emph{strong equilibrium} of an $n$-player cost minimization game $\Gamma$ with strategy profile set $\mathcal{S} = \times_{i =1}^n \mathcal{S}_i$ is an $s \in \mathcal{S}$ such that for all $K \subseteq [n]$ and for all $s_{K}' \in \mathcal{S}_K$ there exists a player $i \in K$ such that,
\begin{equation*}
c_i(s) \leq c_i(s_{K}',s_{-K}),
\end{equation*}
where $c_i$ is the cost function of player $i$ and $(s_{K}', s_{-K})$ denotes the vector obtained from $s$ by replacing the $|K|$ elements at index set $K$ with the elements $s_{K}'$.
\end{definition}
(A \emph{pure Nash equilibrium} is a strategy profile where the above condition is only required for singleton $K$.)
Strong equilibria are guaranteed to exist in almost all network creation games, as we will explain later.


We are interested in determining the \emph{strong price of anarchy} \cite{andelman}, defined as follows.
\begin{definition}
The \emph{strong price of anarchy} of a network creation game $\Gamma$ is the ratio
\begin{equation*}
\text{PoA}(\Gamma) = \max\left\{\frac{C(s)}{C(s^*)} : s \in \text{SE}\right\},
\end{equation*}
where $s^*$ is a \emph{social optimum}, i.e., a strategy profile that minimizes the social cost. Furthermore $\text{SE}$ is the set of strong equilibria of the game.
\end{definition}


A strategy profile $s$ is called \emph{rational} if there is no player pair $i,j \in [n]$ such that $j \in s_i$ and $i \in s_j$. It is clear that all pure Nash equilibria (and thus all strong equilibria) of any network creation game are rational, as are all the social optima. When $s$ is a rational strategy profile, the social cost can be written as $C(s) = \alpha |E(G(s))| + d(s)$, where $E(G(s))$ denotes the edge set of the graph $G(s)$. 

We write $\text{deg}_{G(s)}(i)$ to denote the degree of player $i$ in graph $G(s)$, and we denote by $\text{diam}(G(s))$ the diameter of $G(s)$.
We define the \emph{free-riding} function $f : \mathcal{S} \times [n] \rightarrow \mathbb{N}$ by the formula $f(s, i) = \text{deg}_{G(s)}(i) - |s_{i}|$.
For any strategy profile $s \in \mathcal{S}$ we have the following lower bound for the distance cost of player $i$,
\begin{equation}
  c^{d}_{i}(s) \geq 2n - 2 - \text{deg}_{G(s)}(i). \label{eq:basic0}
\end{equation} 
This follows from the fact that all vertices in the neighborhood of $i$ are at distance $1$, and all vertices not in the neighborhood of $i$ are at distance at least $2$.
Hence, for the overall cost we have
\begin{equation}
  c_{i}(s) \geq 2n - 2 - \text{deg}_{G(s)}(i) + |s_{i}|\alpha = 2n - 2 - f(s,i) + |s_{i}|(\alpha - 1). \label{eq:basic}
\end{equation} 
Moreover, we see that in case $s$ is rational,
\begin{equation}\label{obs_sum}
\sum_{i \in [n]} |s_i| = |E| = \sum_{i \in [n]} f(s,i).
\end{equation}

\paragraph{Graph theory notions.}
We define an \emph{$n$-star} to be a tree of $n$ vertices with diameter $2$, i.e., it is a tree where one vertex is connected to all other vertices. It is straightforward to verify that~(\ref{eq:basic}) is tight when $G(s)$ is an $n$-star, and (more generally) when $G(s)$ has diameter at most $2$. We denote by $K_n$ the complete undirected graph on vertex set $[n]$. We denote by $C_n$ the undirected cycle with $n$ vertices. We denote by $P_n$ the undirected path with $n$ vertices. Lastly, we define the \emph{centroid} of a tree $T = (V,E)$ as the set of vertices $v \in V$ that minimize $\max\{|V_i| : (V_{i},E_{i}) \in \mathcal{C}_{T - v}\}$, where $\mathcal{C}_{T - v}$ denotes the set of connected components of the subgraph of $T$ induced by $V \setminus \{v\}$. In other words, a node $v$ belongs to the centroid of a tree $T$ if it minimizes the size of a maximal connected component the subgraph of $T$ induced by $V \setminus \{v\}$.

\paragraph{Coalitional improvement dynamics.}
A sequence of strategy profiles $(s^{1}, s^{2}, \ldots)$ is called a \emph{path} if for every $k > 1$ there exists a player $i \in [n]$ such that  $s^{k} = (s_{i}', s_{-i}^{k-1})$. We call a path an \emph{improvement path} if it is maximal and for all $k > 1$ it holds that $c_{i}(s^{k}) < c_{i}(s^{k-1})$ where $i$ is the player who deviated from $s^{k-1}$. We say that an improvement path is an \emph{improvement cycle} if there exists a constant $T$ such that $s^{k + T} = s^{k}$ for all $k \ge 1$. A sequence of strategies $(s^{1}, s^{2}, \ldots)$ is called a \emph{best response improvement path} if for all $k > 1$ and all $i$ such that $s^{k}_{i} \neq s^{k - 1}_{i}$ we have $c_{i}(s^{k}) < c_{i}(s^{k-1})$ and there is no $s_i' \in \mathcal{S}_i$ such that $c_i(s_i',s_{-i}^k) < c_{i}(s^k)$ (that is: $s^{k}_{i}$ is a \emph{best response} to $s^{k-1}_{-i}$). A sequence of strategies $(s^{1}, s^{2}, \ldots)$ is called a \emph{coalitional improvement path} if for all $k > 1$ and all $i$ such that $s^{k}_{i} \neq s^{k - 1}_{i}$ we have $c_{i}(s^{k}) < c_{i}(s^{k-1})$.
 
A game has the \emph{(coalitional) finite improvement property ((c-)FIP)} if every (coalitional) improvement path is finite. A game has \emph{finite best response property (FBRP)} if every best response improvement path is finite. We call a game \emph{(c-)weakly acyclic} if for every $s \in \mathcal{S}$ there exists a finite (coalitional) improvement path starting from $s$.
Lastly, we call a network creation game \emph{(c-)weakly acyclic with respect to a class of graphs $\mathcal{G}$} if for every $s \in \mathcal{S}$ such that $G(s) \in \mathcal{G}$, there exists a (coalitional) finite improvement path starting from $s$.
  
  

\section{Structural Properties of Strong Equilibria}
\label{structure}
We provide in this section various results that imply a full characterization of strong equilibria for $\alpha \in (0,2)$, and we resolve a conjecture of \cite{andelman} by showing that any rational strategy $s \in \mathcal{S}$ such that $G(s)$ is a star is a strong equilibrium for all $\alpha \geq 2$. Moreover, we give a family of examples of strategy profiles that form trees of diameter $4$ (hence do not form stars), which we prove to be strong equilibria for $\alpha \geq 2n$. First, for $\alpha \in (0,1)$ the set of strong equilibria is straighforward to derive, as has been pointed out in \cite{andelman}.
\begin{proposition}[\cite{andelman}]\label{slt1} 
For $\alpha < 1$, a strategy profile $s \in \mathcal{S}$ is a strong equilibrium if and only if $s$ is rational and $G(s) = K_n$. 
\end{proposition}
\noindent It is easy to see that the above characterization also holds for the set of Nash equilibria. Hence, strong equilibria and Nash equilibria coincide for $\alpha < 1$.

For $\alpha = 1$, the situation is more complex. Let us first prove the following lemma that holds for all $\alpha < 2$.
\begin{lemma}\label{stw1}
Fix $\alpha < 2$ and suppose that $s \in \mathcal{S}$ is a strong equilibrium. For each sequence of players $(i_{0}, i_{1}, \ldots, i_{k} = i_{0})$ such that $k \ge 3$ in $G(s)$ there exists $t \in \{0, \ldots, k-1 \}$ such that $(i_{t}, i_{t+1}) \in E(G(s))$. In other words, the complement of $G(s)$ is a forest.
\end{lemma}
\begin{proof}
Assume that there exists a sequence $(i_{0}, i_{1}, \ldots, i_{k-1}, i_{k} = i_{0})$ such that $(i_{t}, i_{t+1}) \not \in E(G(s))$, for all $t \in \{0, 1, \ldots, k-1\}$. 
Suppose that for all $t \in \{0,1,\ldots,k-1\}$, player $i_{t}$ buys an additional edge to player $i_{t+1}$. Then the cost of every player $i_{t}$ decreases by at least $2 - \alpha > 0$.
Hence, $s$ is not a strong equilibrium. 
\end{proof}
\noindent Therefore, if $\alpha < 2$ and $s \in \mathcal{S}$ is a strong equilibrium, then there is no independent set of size $3$ in $G(s)$. Also, if $\alpha < 2$ and $|V| \ge 4$, then a strategy profile $s \in \mathcal{S}$, such that $G(s)$ is a star is not a strong equilibrium. Since when $\alpha \in [1,2)$, a rational strategy profile that forms a star is a Nash equilibrium, this implies that the pure Nash equilibria and strong equilibria do not coincide.

In order to characterize the strong equilibria for $\alpha = 1$, we first provide a characterization of the pure Nash equilibria.
\begin{lemma}\label{prop245}
  For $\alpha = 1$, a strategy profile $s \in \mathcal{S}$ is a Nash equilibrium if and only if $s$ is rational and $G(s)$ has diameter at most $2$.
\end{lemma}
\begin{proof} 
First, suppose that $s \in \mathcal{S}$ is a rational strategy profile and that $\text{diam}(G(s)) \ge 3$, i.e., there exists a pair of vertices $i, j$ such that $d_{G(s)}(i,j) \ge 3$. Consider $s'_{i} = s_{i} \cup \{j \}$. Observe that $c^{b}_{i}(s'_{i}, s_{-i}) - c^{b}_{i}(s) = 1$ and $c^{d}_{i}(s'_{i}, s_{-i}) - c^{d}_{i}(s) \le -2$. Hence $c_{i}(s'_{i}, s_{-i}) - c_{i}(s) < 0$ and therefore $s$ is not a Nash equilibrium. Second, suppose that $s \in \mathcal{S}$ is rational and $\text{diam}(G(s)) \le 2$. Take any player $i$ and any strategy $s_{i}' \neq s_{i}$. Let $c^{b}_{i,j}(s) = 1$ if $j \in s_{i}$ and $0$ otherwise. Moreover let $c_{i,j}(s) = d_{G(s)}(i,j) + c_{i,j}^{b}(s)$ be the part of cost $c_{i}(s)$ contributed by vertex $j$. There are two cases to consider. For the first case, $j$ pays for the connection to $i$. Then $c_{i,j}(s) = 1$ and $c_{i,j}(s_{i}', s_{-i}) \ge 1$. For the second case, $j$ does not pay for the connection to $i$. Then $c_{i,j}(s_{i}', s_{-i}) \ge 2$ and $c_{i,j}(s) = 2$  (as either $d_{G(s)}(i,j) = 2$ or $d_{G(s)}(i,j) = 1$ and $c^{b}_{i,j}(s) = 1$). For all $a \in \mathcal{S}$ we have $c_{i}(a) = \sum_{j \in [n] \setminus\{i\}} c_{i,j}(a)$ hence $c_{i}(s'_{i}, s_{-i}) \ge c_{i}(s)$.
\end{proof}
\noindent The following theorem now characterizes the set of strong equilibria for $\alpha = 1$.
\begin{theorem}\label{se1}
  For $\alpha = 1$, a strategy profile $s \in \mathcal{S}$ is a strong equilibrium if and only if $s$ is rational, $G(s)$ has diameter at most $2$, and the complement of $G(s)$ is a forest.
\end{theorem}
\begin{proof}
First, suppose $s \in \mathcal{S}$ is a strong equilibrium. Then $s$ is a pure Nash equilibrium and hence, from Lemma~\ref{prop245}, $G(s)$ has diameter at most $2$. Moreover, from Lemma~\ref{stw1}, it follows that the complement of $G(s)$ is a forest. To prove the converse, supose that $s \in \mathcal{S}$ is rational, $\text{diam}(G(s)) \leq 2$ and the complement of $G(s)$ is a forest. We assume for contradiction. Take any coalition $K$ and assume for contradiction that it can beneficially deviate to any strategy profile $s_{K}' \in \mathcal{S}_K$ for $K$. Let $s' = (s_{K}', s_{-K})$. Fix a player $i \in K$. Since $G(s)$ has diameter at most $2$, from~(\ref{eq:basic}) we have $c_{i}(s) = 2n - 2 - f(s, i)$. Moreover, from~(\ref{eq:basic}) we also have $c_{i}(s') \ge 2n - 2 - f(s', i)$. Hence $f(s', i) > f(s, i)$ for all $i \in K$, which means that for each $i \in K$ there exists a player $j \in K, j \neq i$ such that $i \in s_{j}' \setminus s_{j}$. Therefore there is a cycle in the complement of $G(s)$, which is a contradiction.
\end{proof}

For $\alpha \in (1,2)$, it was shown in \cite{andelman} that strong equilibria do not exist for $n \geq 5$. We present a very simple proof of this in Appendix~\ref{apx:nonexistence}.\footnote{The longer proof of \cite{andelman} however additionally shows that there exists no strategy profile for $n \geq 7$ that is resilient to deviations of sets of $3$ players.}
It can be shown that for $n = 3$ the set of strong equilibria are the rational strategy profiles that form the $3$-star. (Hence, all pure Nash equilibria are strong equilibria in this case.) For $n = 4$ we observe that the only strong equilibria are those that form the cycle on $4$ vertices such that every player buys exactly one edge. See Appendix~\ref{apx:structure} for a proof of these facts. Thus, the following proposition completes our characterization of strong equilibria for $\alpha \in (1,2)$.
\begin{proposition}\label{prop:se12}
Let $\alpha \in (1,2)$ and let $s \in \mathcal{S}$. Then: 
\begin{itemize}
 \item If $n = 3$, strategy profile $s$ is a strong equilibrium if and only if $s$ is rational and $G(s)$ is a $3$-star.
 \item If $n = 4$, strategy profile $s$ is a strong equilibrium if and only if $s$ is rational, $|s_i| = 1$ for all $i$, and $G(s)$ is a cycle. 
 \item If $n \geq 5$, $s$ is not a strong equilibrium. 
\end{itemize}
\end{proposition}

Next, we prove the following conjecture of \cite{andelman}.
\begin{theorem}\label{strongge2}
Let $\alpha \geq 2$ and $s \in \mathcal{S}$. If $s$ is rational and $G(s)$ is a star, then $s$ is a strong equilibrium.
\end{theorem}
\noindent We first prove two lemmas and subsequently prove the theorem. The first of the two lemmas provide bounds on the free-riding function of players who manage to deviate profitably.
\begin{lemma}\label{L2}
Let $\alpha \geq 2$ let $s \in \mathcal{S}$ be a rational strategy profile such that $G(s)$ is a star. Let $K \subseteq [n]$ be a set of players and let $s' = (s'_{K}, s_{-K})$ be a profitable deviation for $K$, i.e., for all $i \in K$, it holds that $c_i(s_K',s_{-K}) < c_i(s)$. Then, for every $i \in K$ such that $\text{deg}_{G(s)}(i) = 1$ we have $f(s', i) > f(s, i)$ and $f(s', i) - f(s, i) \ge |s'_{i}| - |s_{i}| + 1$.
\end{lemma}
Informally, Lemma \ref{L2} shows that if there is a leaf player that participates in a profitable coalitional deviation, 
it must obtain at least 1 more ingoing edge than the number of outgoing edges she wants to pay for, otherwise the deviation is not profitable for her.
\begin{proof}[Proof of Lemma \ref{L2}]
Let $s, s', i, K$ be as in the statement of the lemma. To prove the first inequality observe that from~(\ref{eq:basic}) we have: 
\begin{equation}\label{sz}
2n - 2 + |s'_{i}| (\alpha - 1) - f(s', i) \le c_{i}(s') < c_{i}(s) = 2n - 2 + |s_{i}| (\alpha - 1) - f(s, i). 
\end{equation}
\noindent Suppose $f(s', i) \le f(s, i)$. From~(\ref{sz}) we have
$(|s'_{i}| - |s_{i}|) (\alpha - 1) < f(s', i) - f(s, i) \le 0$.
Since $(\alpha - 1) \ge 1$ we obtain $|s_{i}'| < |s_{i}|$
which is true only if $|s_{i}| = 1$ and $|s_{i}'| = 0$ (as per our assumption that $\text{deg}_{G(s)}(i) = 1$). The fact that $i$ chooses to buy an arc under $s$ in turn implies that $f(s,i) = 0$, so with our assumption $f(s',i) \le f(s,i)$ we infer that $f(s,i) = 0$. But then $i$ is an isolated node under $s'$ so $G(s')$ is disconnected ($\deg_{G(s')}(i) = 0$), which is a contradiction. Hence 
$f(s', i) > f(s, i)$. To prove the second inequality we again use $(|s'_{i}| - |s_{i}|) (\alpha - 1) < f(s', i) - f(s, i)$. 
Since $(\alpha - 1) \ge 1$ and $f(s', i) > f(s, i)$ we get
$|s'_{i}| - |s_{i}| < f(s', i) - f(s, i)$
or equivalently
$|s'_{i}| - |s_{i}| + 1 \le f(s', i) - f(s, i)$.
\end{proof}
\noindent The second lemma bounds the change in the free-riding function for players who do not deviate. Informally, the lemma tells us the following two facts. Firstly, if a star has formed, and deviating coalition of size $k$ consists only of leaf players, then they cannot redirect more than $k$ of their edges, since buying additional edges is too expensive. Secondly, if a deviating coalition contains the center player of the star, then the deviating leaf players outside the coalition  cannot be worse off in terms of a decreased free-riding function, since somebody needs to pay for the edges that ensure the graph remains connected.
\begin{lemma}\label{lemkm1}
Let $\alpha \ge 2$ and let $s \in \mathcal{S}$ be a rational strategy profile such that $G(s)$ is a star. Let $K \subseteq [n]$ be a player set and $s' = (s'_{K}, s_{-K})$ be a profitable deviation for $K$. Then 
\begin{equation*}
\sum_{j \in [n] \setminus K} (f(s', j) - f(s, j)) > -|K|.
\end{equation*}
Moreover, if $K$ contains a vertex $i$ such that $\text{deg}_{G(s)}(i) > 1$, then 
\begin{equation*}
\sum_{j \in [n] \setminus K} (f(s', j) - f(s, j)) \ge 0.
\end{equation*}
\end{lemma}
\begin{proof}
First, suppose there is $i \in K$ such that $\text{deg}_{G(s)}(i) > 1$. In this case, we will show that $f(s', j) \geq f(s, j)$ for all $j \in [n] \setminus K$. To do this, suppose that there is $j \in[n] \setminus K$ such that $f(s', j) < f(s, j)$. Therefore $f(s, j) = 1$ and $f(s', j) = 0$. Since $f(s,j) = 1$ we have $|s_{j}| = 0$. Moreover  
since $j \in [n] \setminus K$ we have $|s_{j}'| = |s_{j}| = 0$. Hence $G(s')$ disconnected and $c_{i}(s') = \infty$, which is a contradiction. Second, suppose $\deg_{G(s)}(i) = 1$ for all $i \in K$. We have $\sum_{i \in K} |s_{i}| \le |K|$. 
Hence $\sum_{j \in [n] \setminus K} (f(s', j) - f(s, j)) \ge -|K|$. If 
$\sum_{j \in [n] \setminus K} (f(s', j) - f(s, j)) = -|K|$, then it must be that under $s$, each player in $K$ buys one arc, and no player buys an arc to any player in $K$. Moreover, under $s$, all arcs bought by players in $K$ are to players in $[n]\setminus K$. Hence, the assumption $\sum_{j \in [n] \setminus K} (f(s', j) - f(s, j)) = -|K|$ implies that the total loss in $f$ among the players in $[n]\setminus K$ is exactly $|K|$, so there is no arc between $[n]\setminus K$ and $K$ in $G(s')$. Thus $G(s')$ is disconnected
and $c_{i}(s') = \infty$ for all $i \in V$, which is also a contradiction. 
\end{proof}

\begin{proof}[Proof of Theorem~\ref{strongge2}]
Let $s \in \mathcal{S}$ be a strategy profile that is rational such that $G(s)$ is a star. It is easy to see that $s$ is a Nash equilibrium (see also \cite{fabrikant}). Suppose that $K \subseteq [n]$ and $s' \in \mathcal{S}_K$ are such that strategy profile $s' = (s'_{K}, s_{-K})$ decreases the costs of all players in $K$. Let $k = |K|$, we have two cases to consider. 

If $\text{deg}_{G(s)}(i) = 1$ for all $i \in K$, then
\begin{eqnarray*}
\sum_{i \in K} (f(s', i) - f(s, i)) & \geq & k + \sum_{i \in K} (|s'_{i}| - |s_{i}|) \\
& = & k + \sum_{i \in [n]} (|s'_{i}| - |s_{i}|) \\
& = & k + \sum_{i \in [n]} (f(s', i) - f(s, i)),
\end{eqnarray*}
where the inequality follows from Lemma~\ref{L2} and the last equality follows from~(\ref{obs_sum}).
Hence $\sum_{i \in [n] \setminus K} (f(s', i) - f(s, i)) \le -k$, which is a contradiction with Lemma~\ref{lemkm1}.

If $K$ contains the center vertex $i$ (i.e., the vertex for which $\text{deg}_{G(s)}(i) > 1$), then
\begin{eqnarray*}
\sum_{j \in K \setminus \{i\}} (f(s', j) - f(s, j)) & \ge & (k - 1)  + \sum_{j \in K \setminus \{i\}} (|s'_{j}| - |s_{j}|) \\
& = & (k - 1) + \sum_{j \in [n]} (|s'_{j}| - |s_{j}|) - (|s'_{i}| - |s_{i}|) \\
& = & (k-1) + \sum_{j \in [n]}(f(s', j) - f(s, j)) - (|s'_{i}| - |s_{i}|),
\end{eqnarray*}
where again the inequality follows from Lemma~\ref{L2} and the last equality follows from~(\ref{obs_sum}).

Since $i$ is a central vertex, we have $c_{i}^{d}(s') \ge c_{i}^{d}(s)$. 
Moreover, $i \in K$, hence $c_{i}(s') < c_{i}(s)$. This implies that $c_{i}^{b}(s') < c_{i}^{b}(s)$ or equivalently $|s'_{i}| < |s_{i}|$. So $\sum_{j \in K \setminus \{i\}} (f(s', j) - f(s, j)) \ge k + \sum_{j \in [n]}(f(s', j) - f(s, j))$. Thus
\begin{eqnarray*}
-k & \ge & \sum_{j \in [n] \setminus K} (f(s', j) - f(s, j)) + (f(s', i) - f(s, i)) \\
& \ge & f(s', i) - f(s, i),
\end{eqnarray*}
where the last inequality follows from Lemma~\ref{lemkm1}.
On the other hand we have $f(s', i) - f(s, i) \ge -(k-1)$, since the change from $s$ to $s'$ could have removed at most $k-1$ edges going to player $i$, which is a contradiction.
\end{proof}

Next, for $\alpha > 2$,  we provide a family of strong equilibria none of which forms a star. More precisely, the graphs resulting from these strong equilibria are trees of diameter $4$.
\begin{example}\label{exa:main}
Our examples are paramatrized by two values $A \in \mathbb{N}$, $A \geq 4$ and $k \in \mathbb{N}$. Let $n = Ak + 2$, and let $\alpha \geq 2n$. We denote player $n$ by $R$. Let $L_1 = \{A, A + 1, \ldots, (A-1)k\}$ and $L_2 = \{(A-1)k + 1, \ldots, n - 1\}$. 
The strategy profile $s^*$ is defined as follows: Player $R$ buys edges to $L_{2}$. Each player in $[A-1]$ buys an edge to $k-1$ players of $L_{1}$ in such a way that every player in $L_{1}$ has degree $1$. Moreover, each player in $[A - 1]$ buys an edge to $R$. Thus, each player in $\{1, \ldots, A - 1\}$ buys $k$ edges, $R$ buys $k+1$ edges, and all the remaining players (i.e., in $L_1$ and $L_2$) buy no edges and are leaves in $G(s^*)$. 
The total number of edges bought by players $\{1, \ldots, A-1, R\}$ is $n-1 = Ak + 1$ so that $G(s^*)$ is a tree.
 Figure~\ref{fig:strong} depicts this strategy profile.
\end{example}

\begin{figure}
  \centering
    \begin{minipage}{\textwidth}
      \centering
        \scalebox{0.5}{\includegraphics{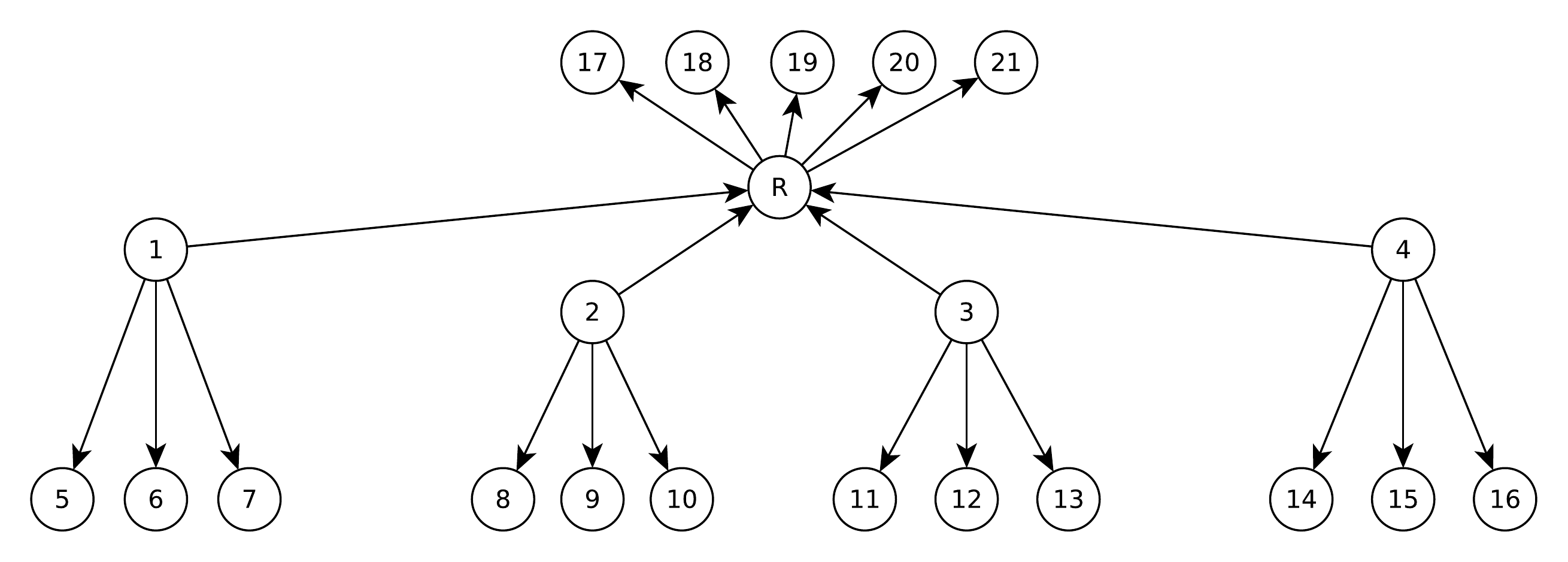}}
     \end{minipage}\hfill 
     \caption{Depiction of the graph $G(s^*)$ formed by the strong equilibrium $s^*$ from example~\ref{exa:main}. The tail of an edge is the player who buys the edge. In the depicted instance of the example we have: $A = 5$, $k = 4$, $n = 22$, $\alpha \ge 44$, $L_{1} = \{5, 6, \ldots, 16 \}$, $L_{2} = \{17, 18, \ldots, 21\}$ and the set of players buying edges is $\{1, \ldots, 4\} \cup \{R\}$. The graph $G(s^*)$ is a tree of diameter $4$. 
     }
     \label{fig:strong}
\end{figure}

Despite that $s^*$ is relatively easy to define, establishing that $s^*$ is a strong equilibrium is challenging.
\begin{theorem}\label{thm:selowerbound}
If $\alpha \geq 2n$, strategy profile $s^*$ forms a non-star tree and is a strong equilibrium. 
\end{theorem}

To prove the theorem, we need to establish some auxiliary lemmas. The outline of the proof is as follows:
\begin{itemize}
\item We first show that every player in any potential deviating coalition cannot buy more edges than she buys under $s^*$. Indeed, under the condition $\alpha \geq 2n$, in this family of examples, buying additional arcs can never be beneficial for any deviating player.
\item As a consequence of the above, and since $s^*$ forms a tree, every deviating player would need to buy exactly the same number of arcs as under $s^*$.
\item We subsequently prove that the center player $R$ can never be in any deviating coalition, and that the maximum distance between a deviating player $i$ and the non-deviating component containing $R$ exceeds $2$. 
\item Finally, we prove that the player $i \in K$, whose distance from non-deviating component containing R is the highest will never decrease his total distance cost by deviating, which establishes the strong equilibrium property.
\end{itemize}

First observe that in $s^*$, there are four different types of node: The root $R$, the players $1, \ldots, A-1$, the leaves $L_1$, and the leaves $L_2$.
The distance costs for each of these types are as follows:
\begin{numcases}{c_i^d(s^*)=}
 (A - 1) + k + 1 + 2(A - 1)(k - 1) = 2n - A - k -2 & \text{ if $i = R$ } \label{eq:r} \\
 k + 2(A - 2 + k + 1) + 3(A - 2)(k - 1) = 3n - A - 3k  -2 & \text{ if $i \in [A-1]$ } \label{eq:a} \\
 1 + 2(A - 1 + k) + 3(A - 1)(k - 1) = 3n - A - k - 4 & \text{ if $i \in L_2$ } \label{eq:l2} \\
 1 + 2(k - 1) + 3(A - 2 + k + 1) + 4(A - 2)(k - 1) = 4n - A - 3k - 4 & \text{ if $i \in L_1$.} \label{eq:l1}
\end{numcases}

 To show that $s^*$ is a strong equilibrium, suppose for contradiction that $K \subseteq [n]$ and $s_K' \in \mathcal{S}_K$ are such that in $s' = (s_K',s^*_{-K})$ it holds that $c_i(s') < c_i(s^*)$ for all $i \in K$.
Under this assumption we first show that no player in $K$ buys more edges under $s'$ than she does under $s^*$.
\begin{lemma}\label{lem:noextraedge}
For all $i \in K$, it holds that $|s_i'| \leq |s^*_i|$.
\end{lemma}
\begin{proof}
 Since $\alpha \geq 2n$ and the minimum possible distance cost of any player under any strategy profile is $n-1$, it holds by~(\ref{eq:r}--\ref{eq:l1}) that 
 \begin{itemize}
  \item For $R$, the distance reduction is at most $c_R^d(s^*) - c_R^d(s') \leq n - A - k - 1 < 2n$, so that $|s_R'| \leq |s^*_R|$.
  \item For $i \in [A-1]$, the distance reduction is at most $c_i^d(s^*) - c_i^d(s') \leq 2n - A - 3k  -1 < 2n$, so that $|s_i'| \leq |s^*_i|$.
  \item For $i \in L_2$, the distance reduction is at most $c_i^d(s^*) - c_i^d(s') \leq 2n - A - k - 3 < 2n$, so that $|s_i'| \leq |s^*_i|$.
  \item For $i \in L_1$, the distance reduction is at most $c_i^d(s^*) - c_i^d(s') \leq 3n - A - 3k -3$. This quantity is not bounded from above by $2n$, but at least it is bounded from above by $4n$, so that $|s_i'| \leq |s^*_i| + 1 = 1$.
 \end{itemize}
Thus, we know that all players except those in $L_1$ do not buy additional edges, and that if a player in $L_1$ buys an (additional) edge, then she buys at most $1$ edge. Thus, it remains to show that also for $i \in L_1$ it holds that $|s_i'| = 0$. Suppose that $|s_i'| \geq 1$. By~(\ref{eq:basic0}), $c_i(s') \geq 2n + c_i^d(s') \geq 4n - 2 - \text{deg}_{G(s')}(i)$, and by~(\ref{eq:l1}), $4n - 2 - \text{deg}_{G(s')}(i) \leq 4n - A - 3k - 2$. Hence, as the cost of $i$ decreases after deviating, it must be that $\text{deg}_{G(s')}(i) \geq A + 3k$. At most $A$ of the edges connected to $i$ are bought by nodes in $[A-1] \cup \{R\}$; none of the edges connected to $i$ are bought by nodes in $L_2$; and one edge connected to $i$ is bought by $i$ herself. Therefore, at least $3k-1$ edges connected to $i$ are bought by nodes in $L_1$. The latter holds for all players $i \in L_1$ that buy an edge. Consider the subgraph of $G(s')$ induced by those nodes of $L_1$ that buy an edge. Direct the edges in this subgraph so that an edge $\{i,j\}$ points from $i$ to $j$ if $j \in s_i$. Since the out-degree of all edges in this subgraph is $1$, there must be a player in the subgraph with in-degree at most $1$, which contradicts that $i$ has an in-degree of at least $3k-1 \geq 2$.
\end{proof}
Since $G(s^*)$ is a tree, it has the minimum number of edges among all connected graphs. Combining this with the lemma above yields that every player buys in $s'$ \emph{exactly} as many edges as in $s^*$.
\begin{corollary}\label{cor:edgesequal}
Graph $G(s')$ is a tree, and for all $i \in [n]$, it holds that $|s_i'| = |s^*_i|$.
\end{corollary}

\begin{lemma}\label{lem:6}
$K \subseteq [A-1]\, .$
\end{lemma}
\begin{proof}
Lemma~\ref{lem:noextraedge} shows that $K \subseteq [A-1] \cup \{R\}$.
Suppose that $R \in K$.
Equation~(\ref{eq:basic0}) implies that $c_R^d(s') \geq 2n - 2 - \text{deg}_{G(s')}(i) \geq 2n - 2 - (k + 1 + A - 1) = 2n - A - k - 2$, which equals $c_R^d(s^*)$ by~(\ref{eq:r}). By Corollary~\ref{cor:edgesequal}, and since by assumption $R$ improves its cost by deviation from $s^*$ to $s'$, it holds that $c_R^d(s') < c_R^d(s^*)$, which is a contradiction. 
\end{proof}

Denote by $L_K = \{j \in L_1\ |\ \exists \, i \in K, j \in s_i \}$ the leaves in $L_1$ that are directly connected to a player in $K$. Let $C_R$ be the connected component of $G(\varnothing,s^*_{-K})$ that contains $R$. Denote by $V(C_R)$ the set of vertices of $C_R$. Let $i \in \arg\max \{d_{G(s')}(i', V(C_R)) : {i' \in K}\}$ be a player in $K$ that has the longest distance to $V(C_R)$ among all players in $K$. (We write $d_{G(s')}(i', V(C_R))$ to denote $\min\{d(i', j) : j \in V(C_R)\}.)$ 
\begin{lemma}\label{lem:7}
We have $d_{G(s')}(i,V(C_R)) \ge 2 \, .$
\end{lemma}
\begin{proof}
Assume for contradiction that the distance of every $i' \in K$ equals $1$.
Since $G(s')$ is a tree, each $i'$ buys exactly $1$ edge to $V(C_R)$. Since the centroid of $C_R$ is $R$, we have
\begin{equation*}
\sum_{j \in V(C_R)} d_{G(s')}(i', j) \geq \sum_{j \in V(C_R)} d_{G(s^*)}(i', j).
\end{equation*}
Moreover, under $s'$, $i'$ has no direct connections to any player of $K$ (otherwise $G(s')$ would not be a tree). The distance from each $i'$ to each of the vertices in $K$ is therefore at least $2$, which is also not smaller than under $s^*$. Lastly, since $i'$ buys $k$ edges in total under $s'$, there are $k-1$ players of $L_K$ at distance $1$. The remainder of the players of $L_K$ are connected to nodes in $K\setminus\{i'\}$, and must lie at distance at least $3$, which is likewise not smaller than under $s^*$. It follows that deviating from $s^*$ to $s'$ is not profitable for $i'$, which is a contradiction.
\end{proof}

\begin{proof}[Proof of Theorem~\ref{thm:selowerbound}]
In $s^*$, the distance from $i$ to $V(C_R)$ is $1$. As $C_R$ has at least $k + 2$ vertices, we see that by deviating from $s^*$ to $s'$, the distance increase of player $i$ to the players of $V(C_R)$ is at least $k+1$. We complete the proof of Theorem~\ref{thm:selowerbound} by showing that by deviating from $s^*$ to $s'$, the distance decrease of player $i$ to the players of $[n]\setminus C_R$ does not exceed $k+1$. This is sufficient, as it implies that $c_i(s') \geq c_i(s^*)$ and thus contradicts the fact that $i \in K$. To see this, observe that in $G(s')$, player $i$ has in his neighborhood at most one player in $K$. If in $G(s')$ there are two or more players in $K$ to which $i$ is directly connected, then one of these players is further away from $V(C_R)$ than $i$ (contradicting the definition of $i$), or there is a cycle in $G(s')$ (contradicting Corollary~\ref{cor:edgesequal}). Let us separately compute the distance improvement to nodes in $L_K$ and to nodes in $K$:
\begin{itemize}
 \item In $G(s^*)$, the distance from $i$ to all $|K|-1$ players in $K$ is $2$, by Lemma~\ref{lem:6}. In $G(s')$ the distance from $i$ to at most one player in $K$ is $1$, while at least $K-2$ player are at distance at least $2$ from $i$. Therefore, the total decrease in distance from $i$ to players in $K$ is at most $1$.
 \item In $G(s^*)$, there are $k-1$ players of $L_K$ at distance $1$ from $i$, and the remaining $|L_K| - k + 1$ players of $L_K$ are at distance $3$ from $i$. In $G(s')$ there are at most $k$ players at distance $1$ from $i$, there are at most $k-1$ players at distance $2$ from $i$ (since the unique player $i'$ of $K$ that is directly connected to $i$ (and buys the edge $(i', i)$) has at most $k-1$ connections to $L_K$). Hence at least $|L_K| - 2k + 1$ players of $L_K$ are at distance at least $3$ from $i$. Therefore, the total decrease in distance from $i$ to players in $L_K$ is at most $(k-1) + (3|L_K| - 3k + 3) - k - (2k - 2) - (3|L_K| - 6k + 3) = k +1$.
\end{itemize}
It follows that by deviating from $s^*$ to $s'$, the maximum possible distance improvement for $i$ to players in $[n]\setminus V(C_R)$ is $k+2$, while, by Lemma~\ref{lem:7}, the distance to at least $k + 2$ vertices of $V(C_R)$ increases by at least $1$. As $|s_i'| = |s_i^*|$ by Corollary~\ref{cor:edgesequal}, the building cost of $i$ is not affected by the deviation, so we conclude that the deviation is not profitable for $i$, which is a contradiction.
\end{proof}

As a last note to this section, we point out that a stronger solution concept of \emph{strict} strong equilibrium has received some attention in the literature. In Appendix \ref{apx:strictstrong}, we briefly discuss our results with respect to strict strong equilibria. 

\section{Bounds on the Strong Price of Anarchy}\label{spoa}
In this section we analyze the strong price of anarchy of network creation games. First, for $\alpha < 2$, we provide exact expressions on the strong price of anarchy using the various insights of Section~\ref{structure}. Subsequently, for higher values of $\alpha$, we provide a sequence of examples that converges to a price of anarchy of $3/2$. This shows that the strong price of anarchy of the complete class of network creation games must lie in the interval $[3/2,2]$, due to the upper bound of $2$ established in \cite{andelman}. We start with $\alpha < 1$, which is trivial.
\begin{proposition}\label{spoal1}
For $\alpha \in (0,1)$, the strong price of anarchy is $1$.
\end{proposition}
(This follows from Proposition~\ref{slt1} and the observation that any rational strategy profile that forms the complete graph minimizes the social cost.) The picture turns out to be relatively complex for $\alpha = 1$.
\begin{theorem}\label{spoa1}
For $\alpha = 1$, the strong price of anarchy is $10/9$ if $n \in \{3,4\}$, and the strong price of anarchy is $(3n + 2)/3n$ if $n \geq 5$.
\end{theorem}
\begin{proof}
By Theorem~\ref{se1}, for $\alpha = 1$ a strategy profile $s$ is a strong equilibrium if and only if it is rational and forms a graph of diameter at most $2$ that is the complement of a forest. This means that vertices connected by an edge are distance $1$ apart, and vertices not connected by an edge are distance $2$ apart. A forest $F = \overline{G(s)}$ has at most $n-1$ edges, so we obtain the following bound on the social cost of a strong equilibrium: 
\begin{equation}\label{eq:forestcomp}
\alpha (n(n-1)/2 - |E(F)|) + 2(2|E(F)| + n(n-1)/2 - |E(F)|) = 3n(n-1)/2 + |E(F)| \leq 3n(n-1)/2 + (n-1).
\end{equation}
This bound is achieved for $n \geq 5$ by taking for $F$ any path on $n$ vertices. Thus, for $\alpha = 1$ and $n \geq 5$, given that the social optimum forms a complete graph, we obtain that the strong price of anarchy is 
\begin{equation*}
\frac{3n(n-1)/2 + (n-1)}{3n(n-1)/2} = \frac{3n(n-1) + 2(n-1)}{3n(n-1)} = \frac{3n+2}{3n},
\end{equation*}
which equals $17/15$ at $n=5$ and decreases monotonically to $1$ as $n$ grows larger.
For $n = 4$, the maximum size forest (such that the complement of it has diameter $2$) has only $2$ edges, and for $n = 3$ it has only $1$ edge. Therefore (using~(\ref{eq:forestcomp})) the strong price of anarchy for $\alpha = 1$ and $n \in \{3,4\}$ equals $10/9$.
\end{proof}

For $\alpha \in (1,2)$, there exists no strong equilibrium if $n \geq 5$ (see \cite{andelman}). Therefore, it remains to derive the strong equilibria for $\alpha \in (1,2)$ and $n \in \{3,4\}$.
\begin{proposition}\label{spoa12}
For $\alpha \in (1,2)$ the strong price of anarchy is $(2\alpha + 8)/(3\alpha + 6)$ if $n = 3$, and the strong price of anarchy is $(4\alpha + 16)/(6\alpha + 12)$ if $n = 4$.
\end{proposition} 
\begin{proof}
For $\alpha < 2$ any social optimum forms the complete graph, so for $n = 3$ the social optimum has social cost $3\alpha + 6$ and for $n = 4$ the social optimum has social cost $6\alpha + 12$.
For $n = 3$, the set of strong equilibria equals the set of pure equilibria, and since the worst case pure equilibria form an $n$-star, we obtain that in this case the strong price of anarchy is $(2\alpha + 8)/(3\alpha + 6)$. For $n = 4$, Proposition~\ref{prop:se12} states that the strong equilibria are rational and form the 4-cycle. The strong price of anarchy for $n = 4$ and $\alpha \in (1,2)$ is therefore $(4\alpha + 16)/(6\alpha + 12)$.
\end{proof}

For $\alpha > 2$ it seems very challenging to prove precise bounds on the strong price of anarchy. However, it is known that for $\alpha \geq 2$ the strong price of anarchy is at most $2$ \cite{andelman}. We now complement this bound by showing that for Example~\ref{exa:main} (given in Section \ref{structure}), the strong price of anarchy is at least $3/2$. 

\begin{theorem}\label{thm:spoa32}
The price of anarchy of network creation games is at least $3/2$.
\end{theorem}

An intuitive sketch for why this theorem holds is as follows. A star graph forms the social optimum, under which the total buying cost is of order $2n^2$, and since almost all pairs of vertices are at distance $2$ from each other, the total distance cost is also of order $2n^2$ under the optimum. Under the strong equilibrium $s^*$ of Example~\ref{exa:main}, the buying cost stays the same, while almost all pairs of vertices are at distance $4$ from each other. Therefore, the ratio between the social costs of the optimum and $s^*$ is roughly $6n^2/4n^2 = 3/2$, which is attained in the limit.

\begin{proof}[Proof of Theorem \ref{thm:spoa32}]
Let $x \geq 4$ and consider the strong equilibrium $s$ given in Example~\ref{exa:main}, for $\alpha = 2n$ and $k = A = x$. The players in $L_1$ each have a distance cost of $4n-4-A-3k = 4x^2 +4 - x - 3x$. Since $|L_1| = (A-1)(k-1) = x^2 - 2x + 1$ the total distance cost of $s$ is at least $4 x^4 - 12 x^3 + 16x^2 - 12 x + 4$. Moreover, $G(s)$ is a tree, so the total building cost of $s$ equals $(n-1)\alpha = (Ak+1)2(Ak+2) = 2x^4 + 6x^2 + 4$. Therefore, the social cost of $s$ satisfies $C(s) \geq 6x^4 - 12x^3 +22x^2 - 12x + 8$.

For $\alpha \geq 2$, the social optimum forms an $n$-star (that follows from~(\ref{eq:basic}) and the fact that star minimizes the number of edges). Thus, the optimal social cost is $(n-1)\alpha + 2(n-1)^2 = 2n(n-1) + 2(n-1)^2 \leq 4n(n-1) = 4x^4 + 12x^2 + 8$.
Combining these two bounds and taking $x$ to infinity, we obtain that the strong price of anarchy is at least:
\begin{equation*}
\lim_{x \rightarrow \infty} \frac{6x^4 - 12x^3 + 22x^2 - 12x + 8}{4x^4 + 12x^2 + 8} = \frac{3}{2} \, . \qedhere
\end{equation*}
\end{proof}

\section{Convergence of Coalitional Improvement Dynamics}
\label{dynamics}
In this section we study the c-FIP and coalitional weak acyclicity of network creation games. 
We state in this section our positive results: c-weak acyclicity holds for $\alpha \in (0,2)$ and for all $\alpha \leq n/2$ in case the starting strategy profile forms a tree.\footnote{Except for $\alpha \in (1,2)$ and $n \geq 5$, in which case we know that strong equilibria do not exist.} On the other hand, our negative results encompass that the c-FIP is not satisfied for any $\alpha$.\footnote{An exception to this is that we can prove that the coalitional finite improvement property is satisfied for the very special case $\alpha > 1$ and $n = 3$.} These negative results are proved by constructing an appropriate improvement cycle for various ranges of $\alpha$. We first show that running best response dynamics on a network creation game ends up in a pure Nash equilibrium.
\begin{lemma}\label{fbrple1}
For $\alpha < 1$, every network creation game has the FBRP.
\end{lemma}
\begin{proof}
We define $\Phi(s) \colon \mathcal{S} \to \mathbb{R}$ 
\begin{equation*}
  \Phi(s) = \sum_{e \in E(G(s))} \mathbf{1}[e \ \text{is in the strategy of exactly one of its endpoints}].
\end{equation*}
Let the strategy profile $s \in \mathcal{S}$ be arbitrary.
We show that $\Phi$ grows when a player $i$ does not play a best response in $s$, and switches to playing a best response.
  
We will construct a strategy $s_{i}' \in \mathcal{S}_{i}$ which we prove to be the unique best response of player $i$ to $s_{-i}$. Let $I_{i} = \{ j \in [n] : (i,j) \not \in E(G(s)) \}$ be the set of players that are not in the neighborhood of $i$ in $G(s)$, and let $A_{i} = \{ j \in [n] : i \in s_{j} \}$ be the set of players that buy an edge to $i$ under $s$. We claim that $s'_{i} = (s_{i} \cup I_{i}) \setminus A_{i}$ is the unique best response to $s_{-i}$. Clearly, it is not a best response if $s'_{i}$ contains a player in $A_i$. Secondly, suppose that $i$ does not include all players of $A_i$ in $s_i'$, then in $(s_i',s_{-i})$ there is a player $j$ at distance $2$ from $i$, so that player $i$ can improve his cost by buying an edge to $j$. This proves that indeed, $s_i'$ is the unique best response to $s_{-i}$.
Moreover we have $\Phi(s_{i}', s_{-i}) - \Phi(s) = |I_{i}| + |A_{i} \cap s_{i}| > 0$, which concludes the proof.
\end{proof}
\noindent From Lemma \ref{fbrple1} and the fact that Nash equilibria and strong equilibria coincide for $\alpha < 1$ (as we also pointed out in Section \ref{structure}), we obtain the following corollary.
\begin{corollary}\label{cwaclt1}
For $\alpha < 1$, every network creation game is c-weakly acyclic.
\end{corollary}

For $\alpha = 1$ we will also prove c-weak acyclicity. To that end, we first prove weak acyclicity.
\begin{lemma}\label{prop246}
For $\alpha = 1$, every network creation game is weakly acyclic.
\end{lemma}
\begin{proof}
Let $s \in \mathcal{S}$ be an arbitrary strategy profile. If $G(s)$ is disconnected then we let an arbitrary player $i \in [n]$ change her strategy to $s_{i}' = [n] \setminus \{i\}$. If $G(s)$ is connected and $\text{diam}(G(s)) \ge 3$ then there are $i, j \in [n]$ such that $d_{G(s)}(i,j) \ge 3$. We set $s_{i}' = s_{i} \cup \{j\}$, which is clearly a change of strategy that decreases the cost of $i$, and we repeat this until we reach a graph $G(s)$ such that $\text{diam}(G(s)) \le 2$. If $s \in \mathcal{S}$ is rational, then we have a Nash equilibrium by Lemma \ref{prop245}. Otherwise there are $i,j$ such that $i \in s_j$ and $j \in s_i$. We set $s_{i}' = s_{i} \setminus \{j\}$ and repeat this until we reach a rational $s$.
\end{proof}

\begin{theorem}\label{cwac1}
For $\alpha = 1$, every network creation game is c-weakly acyclic.
\end{theorem}
\begin{proof}
 Take a strategy profile $s \in \mathcal{S}$. By Lemma \ref{prop246} there exists an improvement path from $s$ leading to a Nash equilibrium $s'$. By Lemma \ref{prop245}, $G(s')$ has diameter at most $2$. Denote by $\overline{G(s')}$ the complement of $G(s')$. If $\overline{G(s')}$ contains a cycle $i_{0}, \ldots, i_{k-1}, i_{k} = i_{0}$, then there exists a coalition $K = \{i_{0}, \ldots, i_{k-1}\}$ that can improve its cost by deviating to the strategy profile $s_K''$, where $s_{t}'' = s_{t}' \cup \{i_{t+1} \}$ for all $t \in K$. Then, every player in $K$ decreases her distance to two other players by $1$, and buys only one additional edge, hence this deviation results in a cost decrease for all players in $K$. We repeat this until there is no cycle left in $\overline{G(s)}$, in which case we reach a strong equilibrium by Theorem \ref{se1}.
\end{proof}
We may also prove that for $\alpha \in (1,2)$ and $n = 3$, network creation games have the c-FIP. Moreover, for the same values of $\alpha$ and $n = 4$, they are c-weakly acyclic. (Recall that for $\alpha \in (1,2)$ and $n \geq 5$, strong equilibria do not exist.)

\begin{proposition}\label{509a}
	For $\alpha > 1$ and $n = 3$ network creation games have the c-FIP.
\end{proposition}
\begin{proof}
	Let $s \in \mathcal{S}$ be an arbitrary strategy profile. To show that the network creation game has the c-FIP, we claim that $\Phi(s) = \sum_{i \in [n]} (c_{i}^{d}(s) + 2 c_{i}^{b}(s))$ is a function that decreases with every profitable deviation made by any coalition. Observe that if $G(s)$ is disconnected, then any profitable deviation leads to a connected graph. From now on we assume that $G(s)$ is a connected graph. Suppose there is an improving coalition $K \subseteq [n]$ that deviates to $s' = (s_{K}', s_{-K})$. Because there are only $3$ players, it is straightforward to verify that $|s_i'| \le |s_i|$ for all $i \in [n]$. Observe that for each vertex $j \not \in K$, we have $c_{j}(s') - c_{j}(s) \le 1$. Moreover if $c_{j}(s') - c_{j}(s) = 1$, then $c_{i}^{d}(s') - c_{i}^{d}(s) = 1$ for some $i \in K$ and since $K$ was an improving coalition we have $c_{i}(s') - c_{i}(s) < 0$ and hence $c_{i}^{b}(s') - c_{i}^{b}(s) \le -\alpha$. Therefore, $\Phi(s') - \Phi(s) \le 2 - 2 \alpha < 0$, which concludes the proof.
\end{proof}

\begin{proposition}\label{510a}
	For $\alpha \in (1,2)$ and $n = 4$, network creation games are c-weakly acyclic.
\end{proposition}
\begin{proof}
	Let $s \in \mathcal{S}$ be any strategy profile. Without loss of generality we can assume that $s$ is rational and $G(s)$ is connected.
	There are only $6$ connected graphs of $4$ vertices (up to isomorphism). 
	We consider all the cases. 
	If $G(s) = C_{4}$ and $|s_i| = 1$ for all $i \in [n]$ then $s$ is a strong equilibrium by Lemma~\ref{510}.
	If $G(s) = C_{4}$ and $|s_i| = 2$ for some $i \in [n]$, then $i$ can profitably deviate in such a way that the 4-star is formed. From the 4-star, a coalitional improvement path to a strong equilibrium exists, as is shown in Figure~\ref{wsio}, where also for all remaining strategy profiles an improvement path is shown.
	Improving coalitions in the remaining five cases are shown in Figure~\ref{wsio}. \qedhere
	
	\begin{figure}
		\centering
		\begin{minipage}{0.4\textwidth}
			\raggedleft
			\scalebox{0.40}{\includegraphics{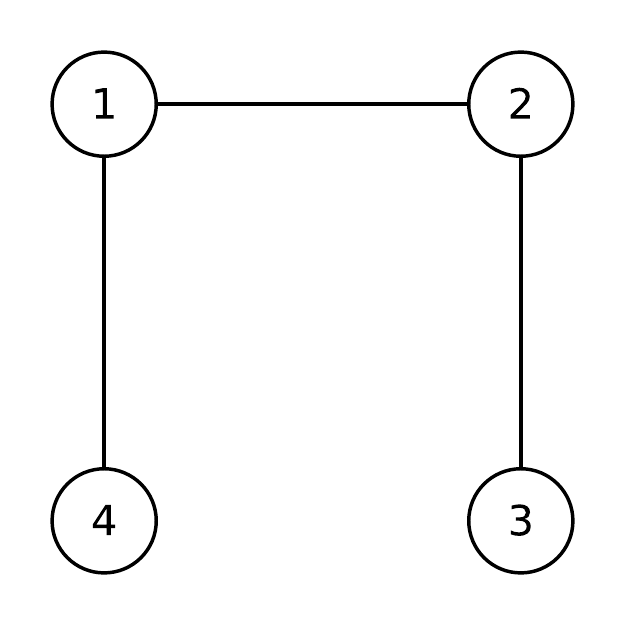}}
		\end{minipage}\hfill 
		\begin{minipage}{0.4\textwidth}
			\raggedright
			\scalebox{0.40}{\includegraphics{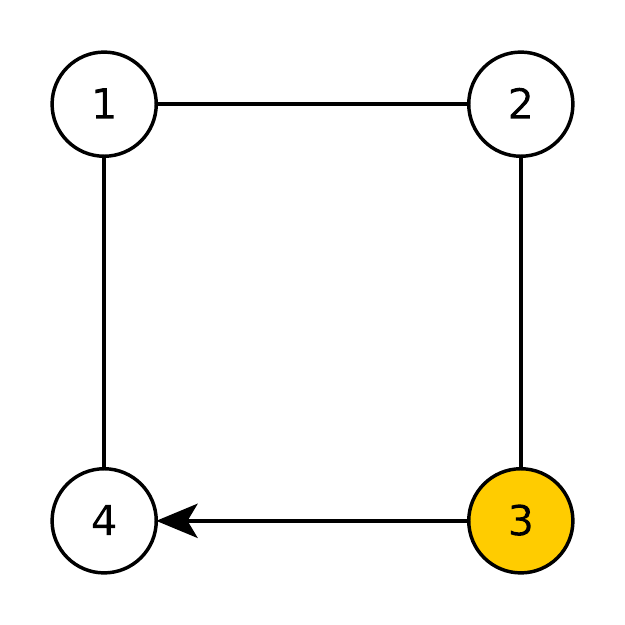}}
		\end{minipage}\hfill 
		\begin{minipage}{0.4\textwidth}
			\scalebox{0.40}{}
		\end{minipage}\hfill
		\begin{center} $P_{4} \to C_4$.\end{center}\label{rys linia}
		
		\centering
		\begin{minipage}{0.4\textwidth}
			\raggedleft
			\scalebox{0.40}{\includegraphics{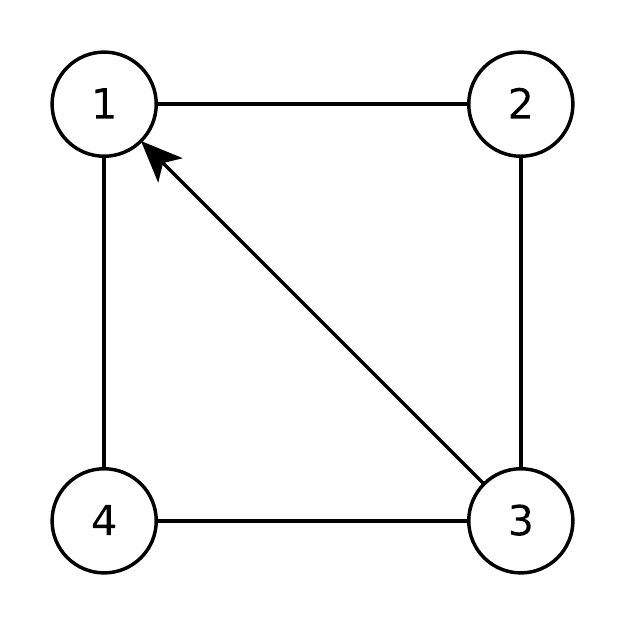}}
		\end{minipage}\hfill
		\begin{minipage}{0.4\textwidth}
			\raggedright
			\scalebox{0.40}{\includegraphics{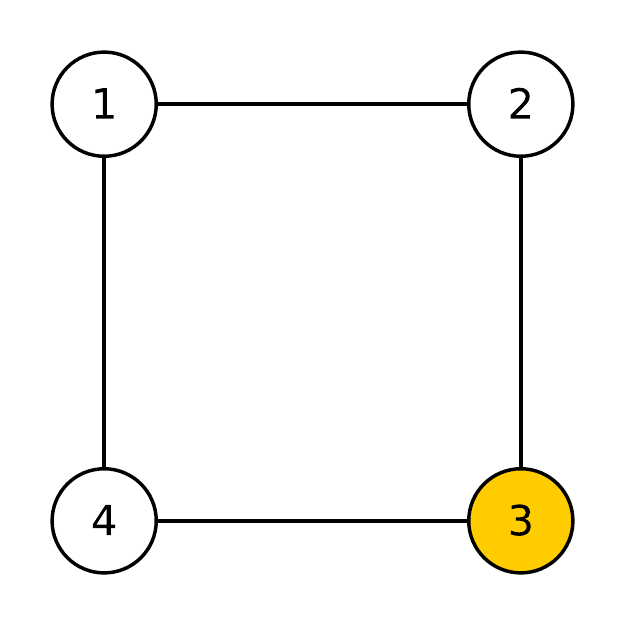}}
		\end{minipage}\hfill
		\begin{minipage}{0.4\textwidth}
			\scalebox{0.45}{}
		\end{minipage}\hfill
		\begin{center}$K_{4} - e \to C_{4}$.\end{center}\label{K4m}
		
		\centering
		\begin{minipage}{0.4\textwidth}
			\raggedleft
			\scalebox{0.40}{\includegraphics{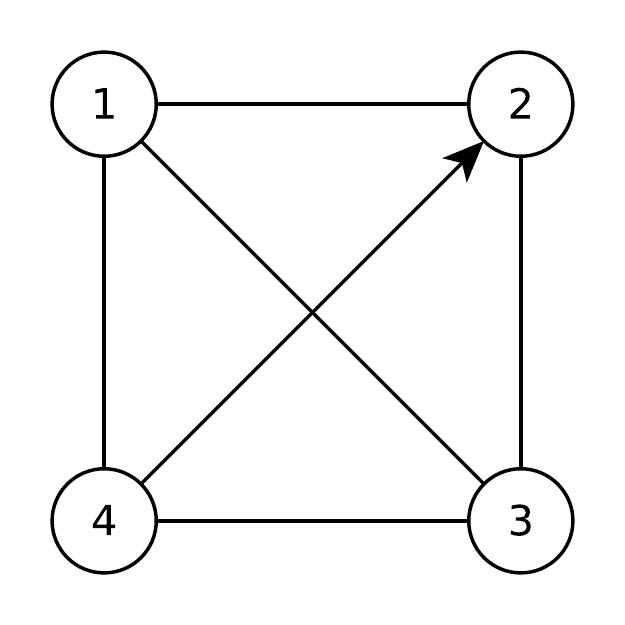}}
		\end{minipage}\hfill
		\begin{minipage}{0.4\textwidth}
			\raggedright
			\scalebox{0.40}{\includegraphics{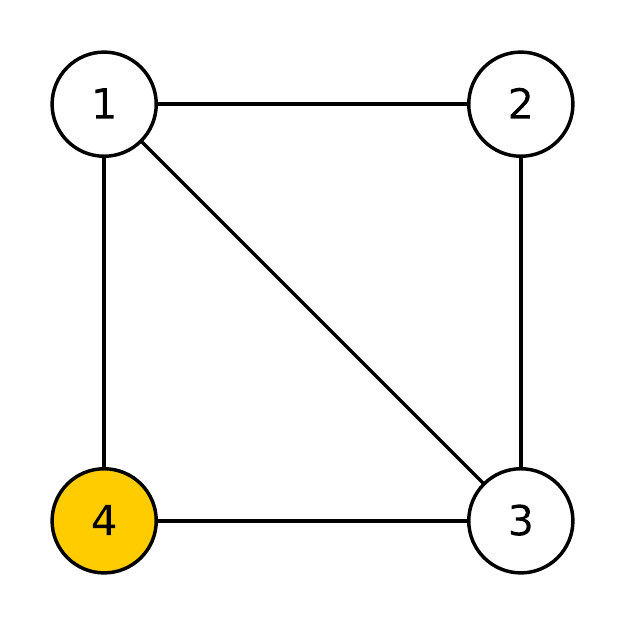}}
		\end{minipage}\hfill
		\begin{minipage}{0.4\textwidth}
			\scalebox{0.40}{}
		\end{minipage}\hfill
		\begin{center}$K_{4} \to K_{4} - e$.\end{center}\label{K4}
		
		\centering
		\begin{minipage}{0.4\textwidth}
			\raggedleft
			\scalebox{0.40}{\includegraphics{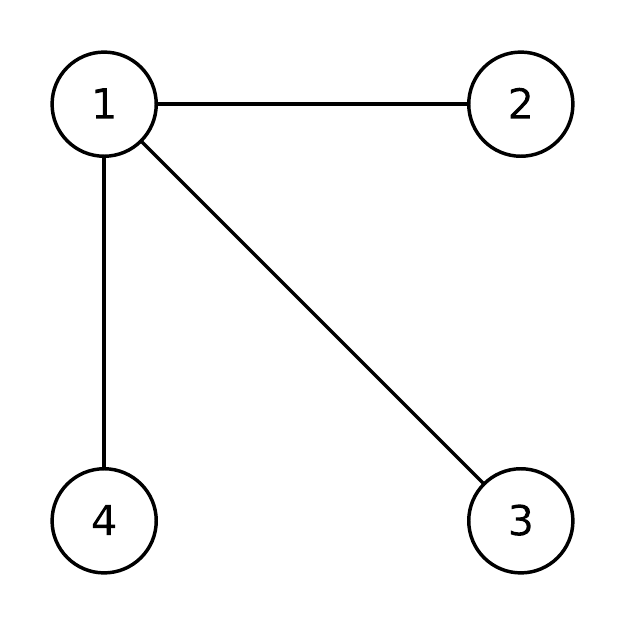}}
		\end{minipage}\hfill 
		\begin{minipage}{0.4\textwidth}
			\raggedright
			\scalebox{0.40}{\includegraphics{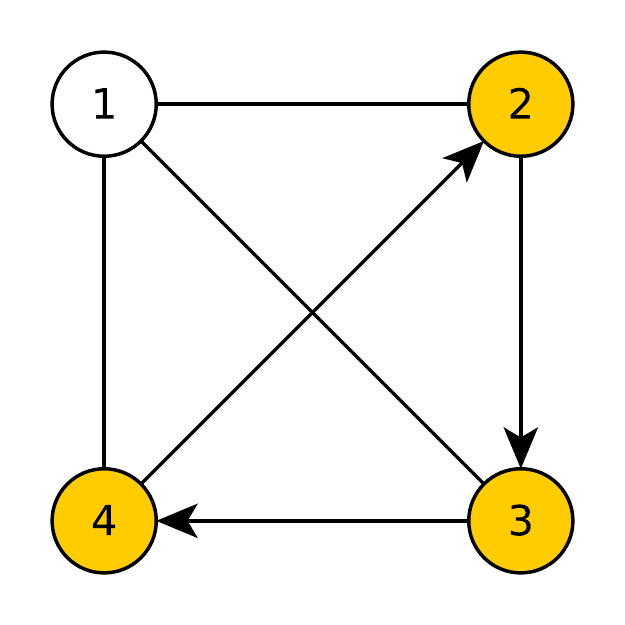}}
		\end{minipage}\hfill 
		\begin{minipage}{0.4\textwidth}
			\scalebox{0.40}{}
		\end{minipage}\hfill
		\begin{center}4-star $\to K_{4}$.\end{center}\label{K13}
		
		\centering
		\begin{minipage}{0.4\textwidth}
			\raggedleft
			\scalebox{0.40}{\includegraphics{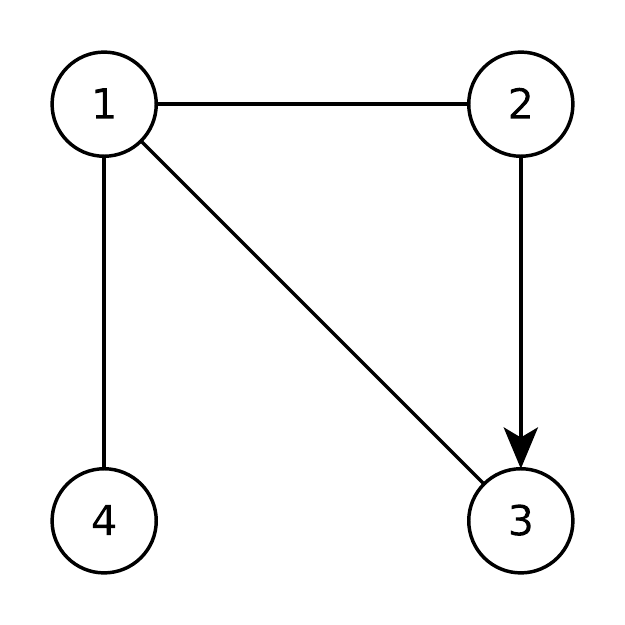}}
		\end{minipage}\hfill
		\begin{minipage}{0.4\textwidth}
			\raggedright
			\scalebox{0.40}{\includegraphics{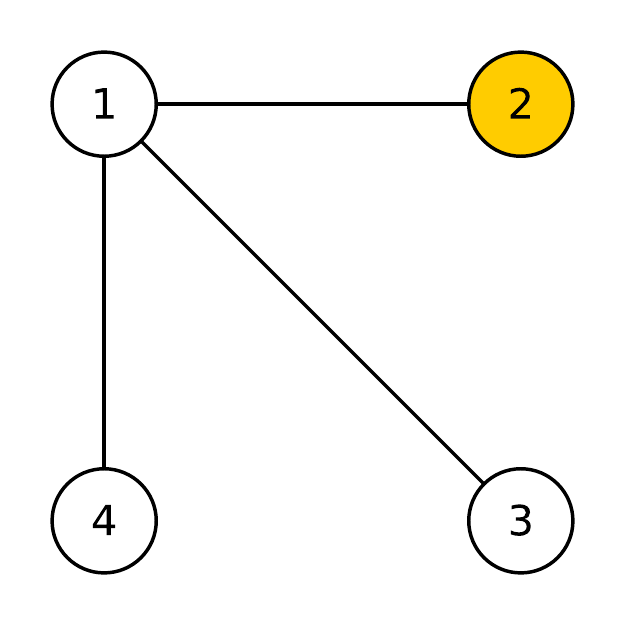}}
		\end{minipage}\hfill
		\begin{minipage}{0.4\textwidth}
			\scalebox{0.40}{}
		\end{minipage}\hfill
		\begin{center}$K_{3} \cup P_{1} \to$ 4-star.\end{center}\label{rys ost}
		
		\caption{Improving coalitions in case of $\alpha \in (1,2)$ and $n = 4$. The pictures in the left column depict the starting strategy profiles, while the right column shows the strategy profiles after deviation. Yellow-colored players are those who changed their strategy. The direction of an edge indicates who buys it (the buyer is the tail of the edge). In case an edge is displayed undirected, the identity of the buyer of the edge is irrelevant due to isomorphism.}\label{wsio}
	\end{figure}      
\end{proof}

For $\alpha \leq n/2$ we now prove that c-weak acyclicity is satisfied as long as our starting strategy profile forms a tree. We first need a preliminary, well known, result about a centroid of a trees.
\begin{lemma}[\cite{kang1975}]\label{lem:centroid}
	Let $T = (V,E)$ be a tree, and let $v \in V$ belong to the centroid of $T$. It holds that $\max\{|V_i| : (V_{i},E_{i}) \in \mathcal{C}_{T - v}\} \le (1/2)|V| \,$.
\end{lemma}
\begin{proof}
	Suppose for the sake of contradiction $\max_{(V_{i},E_{i}) \in \mathcal{C}_{T - v}} |V_{i}| > (1/2)|V|$. Let $u$ be the vertex adjacent to $v$ that belongs to the biggest connected component of $\mathcal{C}_{T - v}$, denoted by $C_{u}$. We can divide $\mathcal{C}_{T-u}$ into two parts: the connected component containing $v$, and the rest. By assumption, the connected component that contains $v$ consist of at most $(1/2)|V|$ vertices. Furthermore, the remaining connected components of $\mathcal{C}_{T-u}$ are strictly contained in $C_{u}$. Therefore their sizes are strictly less than the size of $C_{u}$. This means that $v$ does not belong to a centroid, which is a contradiction.
\end{proof}

\begin{proposition}\label{wat}
For $\alpha \in [2, n/2)$, let $s \in \mathcal{S}$ be such that $G(s)$ is a tree. Then there exists an improvement path resulting in a strong equilibrium. Hence, every network creation game is weakly acyclic and c-weakly acyclic with respect to trees.
\end{proposition}
\begin{proof}
Let $s \in \mathcal{S}$ and suppose $G(s)$ is a tree. Let $v \in [n]$ belong to the centroid of $G(s)$. Consider the following sequence of deviations. If there is a player $i$ such that $d_{G(s)}(i, v) \ge 2$, then $s_{i}' = s_{i} \cup \{v\}$ and $s' = (s_{i}', s_{-i})$. Repeat this step with $s = s'$ until $d_{G(s)}(i, v) = 1$ for all $i \in V \setminus \{v\}$. Observe that since $v$ is belongs to the centroid of $G(s)$, by Lemma \ref{lem:centroid}, player $i$ decreases the distance to at least $n/2$ players by at least $1$ by buying an edge to $v$. This exceeds the cost of $\alpha$, hence this deviation is profitable.

If there is no player $i$ such that $d_{G(s)}(i,v) \geq 2$, and $G(s)$ is not a star, then there are players $i, j \in [n]$ such that $i \neq v$, $j \neq v$ and $j \in s_{i}$. Let $s_{i}' = s_{i} \setminus \{j\}$. Repeat this step until $G(s)$ is a star. Observe that player $i$ is better off by this strategy change. She saves $\alpha > 1$ in her building cost and her distance cost increases by only $1$, since for each player not in $i$'s neighborhood there is a shortest path through $v$. Hence the only loss is the distance increase between $i$ and $j$.

If $s$ is rational after this sequence of deviations, then we have reached a strong equilibrium by Theorem \ref{strongge2}. Otherwise there are $i,j$ such that $i \in s_j$ and $j \in s_i$. We set $s_{i}' = s_{i} \setminus \{j\}$ and repeat this step until we reach a rational $s$.
\end{proof}

The remaining lemmas of this section show that for all $\alpha$, network creation games do not have the c-FIP.
\begin{lemma}\label{cfiplt1}
	For $\alpha < 1$, no network creation game has the c-FIP.
\end{lemma}
\begin{proof}
	Fix $n \ge 3$ and $s \in \mathcal{S}$ (see Figure~\ref{rys cykl2}) such that
	\begin{equation*}
	s_{1} = \{2\}, s_{2} = \emptyset, s_{3} = \{1\} \, .
	\end{equation*}
	Moreover, if $n \ge 4$, then we set 
	\begin{equation*}
	s_{4} = \ldots = s_{n} = \{1, 2 ,3\}.
	\end{equation*}
	We have 
	\begin{equation*}
	c_{1}(s) = (n - 1) + \alpha,c_{2}(s) = n, c_{3}(s) = n + \alpha.
	\end{equation*}
	Take the coalition $K = \{2, 3\}$ with the strategy deviation $s_{2}' = \{3\}, s_{3}' = \emptyset$. We obtain
	\begin{equation*}
	c_{1}(s') = n + \alpha,c_{2}(s') = (n - 1) + \alpha, c_{3}(s) = n,
	\end{equation*}
	so $K$ is an improving coalition. But the resulting strategy profile is isomorphic to the previous one, so the game has an infinite improvement path or equivalently does not have the c-FIP. \qedhere
	
	\begin{figure}
		\centering
		\begin{minipage}{0.4\textwidth}
			\raggedleft
			\scalebox{0.6}{\includegraphics{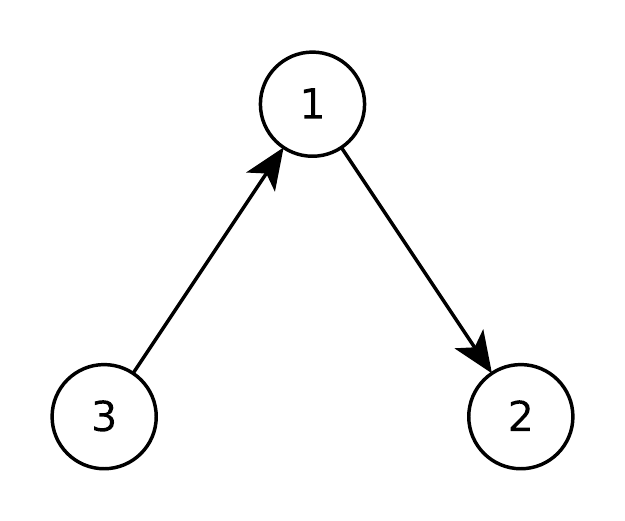}}
		\end{minipage}\hfill 
		\begin{minipage}{0.4\textwidth}
			\raggedright
			\scalebox{0.6}{\includegraphics{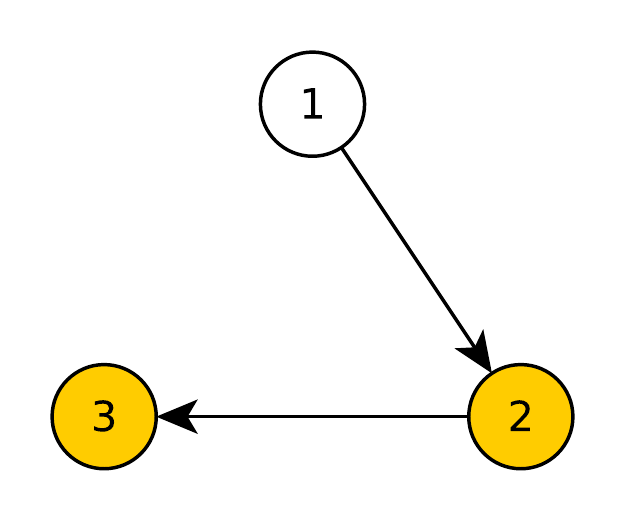}}
		\end{minipage}\hfill 
		
		\caption{Coalitional improvement cycle for $\alpha \in (0,1)$ and $n = 3$. Yellow-colored players are those who changed their strategy. The direction of an edge indicates who buys it (the buyer is the tail of the edge).}\label{rys cykl2}
	\end{figure}
\end{proof}

\begin{lemma}\label{prop-cFIP1}
	For $\alpha = 1$, there exists a network creation game that does not have the c-FIP.
\end{lemma}
\begin{proof}
	Let $n = 13$ and let $s \in \mathcal{S}$ be the strategy profile depicted in Figure~\ref{cFIP1}. The costs of the first seven players are shown in Table~\ref{tabelka1}. Observe that $c_{i}(s) > c_{i+1}(s)$ for $i \in \{1, \ldots, 6\}$. Hence if a player $i \in \{1, \ldots, 6\}$ could somehow deviate to ``take the role'' of player $i+1$, then he would do so. The following coalitional deviation does exactly that. Let $K = \{1,\ldots,6\}$ and define for $i \in K$ the strategy $s_i' = \{\sigma(j) : j \in s_{i+1}\}$, where $\sigma$ is a permutation on $[n]$ such that $\sigma(1) = 7$, $\sigma(i) = i-1$ for $2 \leq i \leq 7$ and $\sigma(i) = i$ for $i \geq 8$. Hence the strategy profile $s' = (s_{K}, s_{-K})$ is a deviation where all players in $K$ reduce their cost. Moreover, $G(s)$ and $G(s')$ are isomorphic (observe that player 7 in $s'$ now ``has the role'' of player $1$ in $s$), which implies that there is an improvement cycle.\qedhere
	\begin{table}
		\centering
		\small
		\begin{tabular}{|c@{\hskip 10pt}|c@{\hskip 10pt}|c@{\hskip 10pt}|c@{\hskip 10pt}|c@{\hskip 10pt}|c@{\hskip 10pt}|c@{\hskip 10pt}|c@{\hskip 10pt}|c@{\hskip 10pt}|}
			\hline 
			Player&1&2&3&4&5&6&7\\
			\hline
			$c(s)$&36&27&26&25&24&23&21\\
			\hline 
			$c(s')$&27&26&25&24&23&21&36\\
			\hline 
		\end{tabular}
		\vspace*{0.1cm}
		\caption{Costs of players 1 to 7 under strategy profiles $s$ and $s'$, case $\alpha = 1$.}\label{tabelka1}
	\end{table}
	
	\begin{figure}
		\centering
		\begin{minipage}{0.5\textwidth}
			\centering
			\scalebox{0.6}{\includegraphics{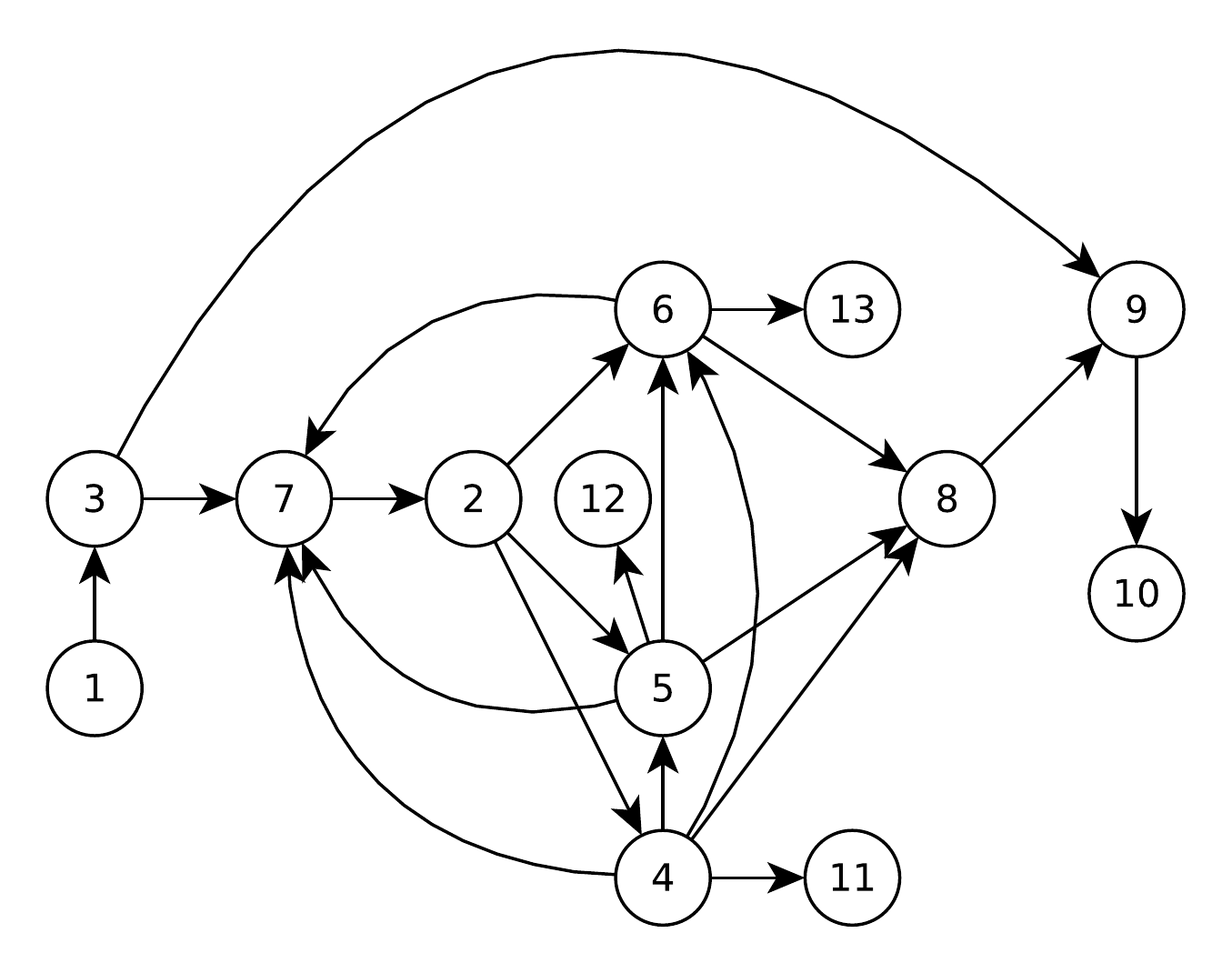}}
		\end{minipage}\hfill 
		\caption{The strategy profile $s$. The direction of an edge indicates who buys it (the buyer is the tail of the edge).}\label{cFIP1}
	\end{figure}
\end{proof}

\begin{proposition}\label{cfip12}
	For $\alpha \in (1,2)$, network creation games do not have the c-FIP.
\end{proposition}
\begin{proof}
	Let $n = 4$, and let $s \in \mathcal{S}$ be a strategy profile such that $G(s)$ is a $4$-star. Let $1$ be the vertex of the star with degree 3. A sequence of strategies which is an improvement cycle of the game is depicted in Figure~\ref{rys cykl}. \qedhere
	
	\begin{figure}
		\centering
		\begin{minipage}{0.3\textwidth}
			\centering
			\scalebox{0.6}{\includegraphics{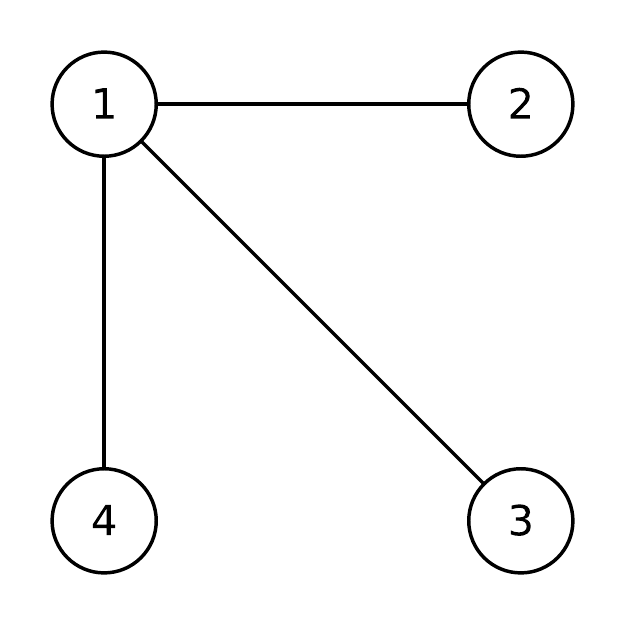}}
		\end{minipage}\hfill 
		\begin{minipage}{0.3\textwidth}
			\centering
			\scalebox{0.6}{\includegraphics{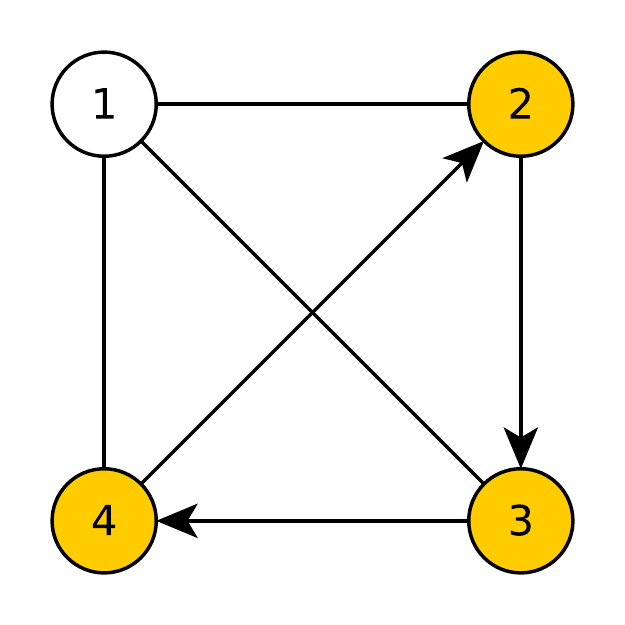}}
		\end{minipage}\hfill 
		\begin{minipage}{0.3\textwidth}
			\centering
			\scalebox{0.6}{\includegraphics{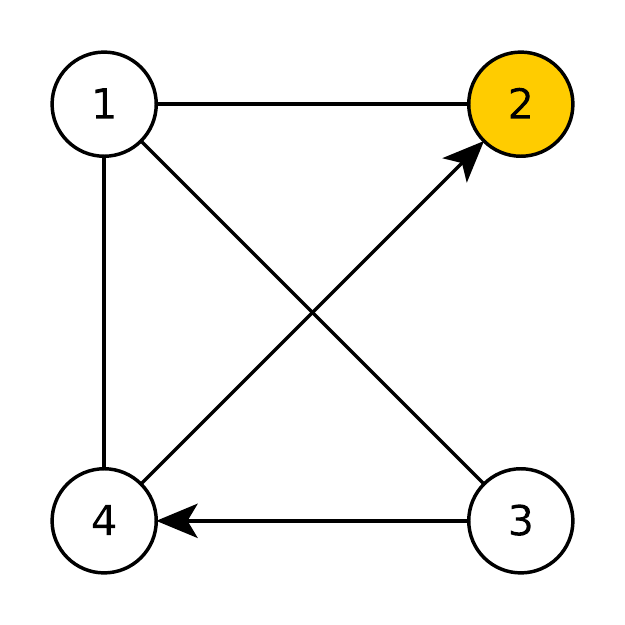}}
		\end{minipage}\hfill 
		\begin{minipage}{0.3\textwidth}
			\centering
			\scalebox{0.6}{\includegraphics{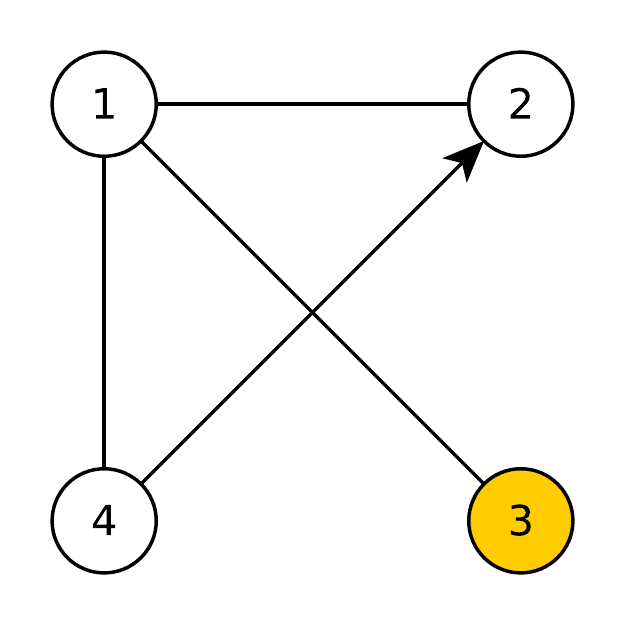}}
		\end{minipage}\hfill 
		\begin{minipage}{0.3\textwidth}
			\centering
			\scalebox{0.6}{\includegraphics{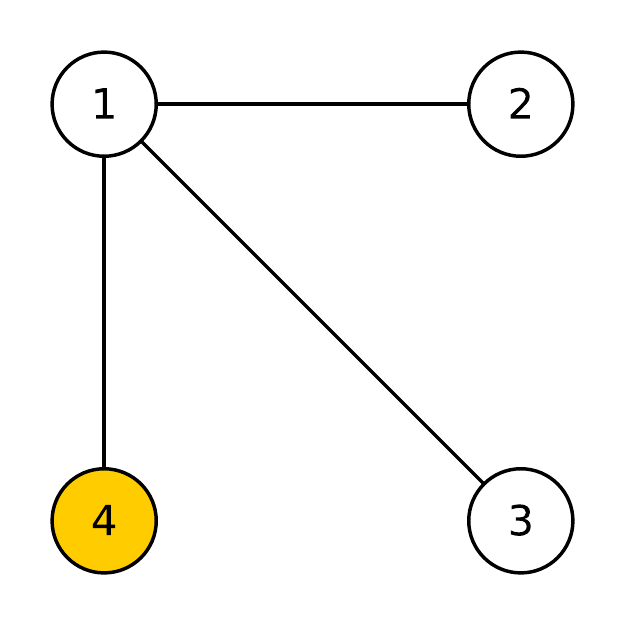}}
			%
		\end{minipage}\hfill
		\begin{minipage}{0.3\textwidth}
			\scalebox{0.6}{}
		\end{minipage}\hfill
		\caption{Coalitional improvement cycle for $\alpha \in (1,2)$ and $n = 4$. Yellow-colored players are those who changed their strategy. The direction of an edge indicates who buys it (the buyer is the tail of the edge). In case an edge is displayed undirected, the identity of the buyer of the edge is irrelevant due to isomorphism.}\label{rys cykl}
	\end{figure}
\end{proof}

\begin{lemma}\label{prop-cFIP2}
	For $\alpha = 2$, there exists a network creation game that does not have the c-FIP.
\end{lemma}
\begin{proof}
	The proof is very similar to the proof of Lemma~\ref{prop-cFIP1}. 
	Let $n = 13$ and let $s \in \mathcal{S}$ be the strategy profile depicted in Figure~\ref{cFIP2}. The costs of the first seven players are shown in Table~\ref{tabelka2}. Observe that $c_{i}(s) > c_{i+1}(s)$ for $i \in \{1, \ldots, 6\}$. The rest of the proof is the same as in the proof of Lemma~\ref{prop-cFIP1}.\qedhere
	
	\begin{table}
		\centering
		\small
		\begin{tabular}{|c@{\hskip 10pt}|c@{\hskip 10pt}|c@{\hskip 10pt}|c@{\hskip 10pt}|c@{\hskip 10pt}|c@{\hskip 10pt}|c@{\hskip 10pt}|c@{\hskip 10pt}|c@{\hskip 10pt}|}
			\hline 
			Player&1&2&3&4&5&6&7\\
			\hline
			$c(s)$&41&33&32&31&29&27&25\\
			\hline 
			$c(s')$&33&32&31&29&27&25&41\\
			\hline 
		\end{tabular}
		\vspace*{0.1cm}
		\caption{Costs of players 1 to 7 under strategy profiles $s$ and $s'$, case $\alpha = 2$.}\label{tabelka2}
	\end{table}
	
	\begin{figure}
		\centering
		\begin{minipage}{0.5\textwidth}
			\centering
			\scalebox{0.6}{\includegraphics{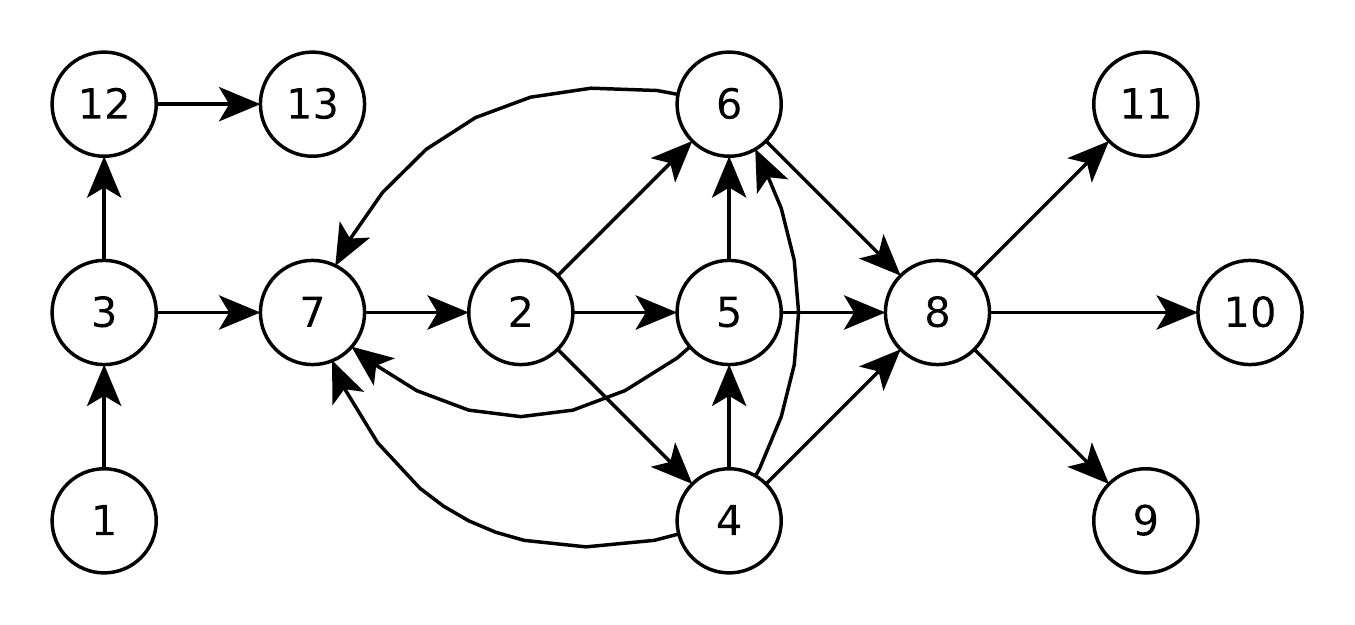}}
		\end{minipage}\hfill 
		\caption{The strategy profile $s$. The direction of an edge indicates who buys it (the buyer is the tail of the edge).}\label{cFIP2}
	\end{figure}
\end{proof}

Lastly, for $\alpha > 2$, the example used in Theorem 1 of \cite{Brandes2008} implies that network creation games are not potential games.\footnote{It may be of independent interest to establish for which choices of $\alpha$ network creation games are potential games. To do so we can make use of the fact that a game has the $FIP$ if and only if it is a generalized ordinal potential game, see \cite{andelman}. From Proposition \ref{509a} we know that network creation games with $\alpha > 1$ and $n = 3$ do have the $FIP$, and therefore they are generalized ordinal potential games. Moreover, the result from \cite{Brandes2008} that shows that network creation games are not (generalized ordinal) potential games, does not only hold for $\alpha > 2$, but also for $\alpha \in (1,3/2)$. For the remaining range of network creation games (i.e., with $\alpha \in [0,1]$ and $\alpha \in [3/2,2], n > 4$) this question remains unresolved.} Hence they do not possess the FIP and the c-FIP for this range of $\alpha$.
\begin{corollary}\label{cor:cfip2}
	For $\alpha > 2$, network creation games do not have the FIP and the c-FIP.
\end{corollary}

\section{Discussion}\label{sec:discussion}
The aim of this paper has been to contribute to the understanding the structure and quality of strong equilibria in network creation games, and secondly to understand the associated improvement dynamics that reach such equilibria. Related to our set of structural and price of anarchy results, we mention the following interesting questions that we leave open for future research: 
\begin{itemize}
 \item What is the exact strong price of anarchy of the class of network creation games? Our work shows that it must lie in the interval $[3/2,2]$.
 \item Does there exist a non-star strong equilibrium for $\alpha \in (2, 2n)$? 
 \item Do there exist strong equilibria that form trees of arbitrarily high diameter, and do there exist strong equilibria that are not trees?
\end{itemize}

With regard to our results on improvement dynamics, the main question that we would like to see answered is whether c-weak acyclicity holds in general, i.e., whether from every strategy profile there exists a coalitional improvement path to a strong equilibrium. Our current results show that this holds when the starting strategy profile forms a tree.

\section*{Acknowledgments}
The first author was partially supported by the NCN grant 2014/13/B/ST6/01807 and the second author by the NWO grant 612.001.352, and by the EPSRC Grant EP/P020909/1. We thank Krzysztof R. Apt for useful suggestions concerning the organization of this paper and some of its topics. Moreover, we thank Mateusz Skomra for various helpful discussions and feedback. This research has been done while the first author was a MSc student at the University of Warsaw and the second author was a postdoctoral researcher at Centrum Wiskunde \& Informatica in the Networks and Optimization group. The final version of the paper was also partially written while the second author was a postdoctoral researcher at University of Liverpool.

\bibliography{networkcreation}

\begin{thebibliography}{10}

\bibitem{albers}
S.~Albers, S.~Eilts, E.~Even-Dar, Y.~Mansour, and L.~Roditty.
\newblock On {N}ash equilibria for a network creation game.
\newblock In {\em Proceedings of the 17th Symposium on Discrete Algorithms
  (SODA)}, pages 89--98. SIAM, 2006.

\bibitem{alon}
N.~Alon, E.~D. Demaine, M.~T. Hajiaghayi, and T.~Leighton.
\newblock Basic network creation games.
\newblock {\em SIAM Journal on Discrete Mathematics}, 27(2):656--668, 2013.

\bibitem{messegue}
A.~{\`{A}}lvarez and A.~Messegu{\'{e}}.
\newblock Selfish network creation with non-uniform edge cost.
\newblock {\em ArXiv}, 1706.09132, 2017.

\bibitem{alvarez2018constant}
C.~{\`A}lvarez and A.~Messegu{\'e}.
\newblock On the constant price of anarchy conjecture.
\newblock {\em arXiv}, 1809.08027, 2018.

\bibitem{andelman}
N.~Andelman, M.~Feldman, and Y.~Mansour.
\newblock Strong price of anarchy.
\newblock {\em Games and Economic Behavior}, 65(2):289--317, 2009.

\bibitem{SEdef}
R.~J. Aumann.
\newblock Acceptable points in general cooperative $n$-person games.
\newblock In R.~D. Luce and A.~W. Tucker, editors, {\em Contribution to the
  theory of games, Volume IV (Annals of Mathematical Study 40)}, pages
  287--324. Princeton University Press, 1959.

\bibitem{Avrachenkov2016b}
K.~Avrachenkov, G.~Neglia, and V.~Singh.
\newblock Network formation games with teams.
\newblock {\em Journal of Dynamics {\&} Games}, 3(4):17, 2016.

\bibitem{Avrachenkov2016}
K.~Avrachenkov and V.~V. Singh.
\newblock Stochastic coalitional better-response dynamics and stable
  equilibrium.
\newblock {\em Automation and Remote Control}, 77(12):2227--2238, 2016.

\bibitem{balagoyal}
V.~Bala and S.~Goyal.
\newblock A noncooperative model of network formation.
\newblock {\em Econometrica}, 68(5):1181--1229, 2000.

\bibitem{balagoyal2}
V.~Bala and S.~Goyal.
\newblock A strategic analysis of network reliability.
\newblock {\em Review of Economic Design}, 5(3):205--228, 2000.

\bibitem{billand}
P.~Billand, C.~Bravard, and S.~Sarangi.
\newblock Existence of {N}ash networks in one-way flow models.
\newblock {\em Economic Theory}, 37(3):491--507, 2008.

\bibitem{lenznerultranew}
D.~Bil{\`{o}} and P.~Lenzner.
\newblock On the tree conjecture for the network creation game.
\newblock {\em ArXiv}, 1710.01782, 2017.

\bibitem{Brandes2008}
U.~Brandes, M.~Hoefer, and B.~Nick.
\newblock Network creation games with disconnected equilibria.
\newblock In {\em Internet and Network Economics: 4th International Workshop,
  WINE 2008 (Proceedings)}, volume 5385 of {\em Lecture Notes in Computer
  Science}, pages 394--401. Springer, 2008.

\bibitem{lenzner4}
A.~Chauhan, P.~Lenzner, A.~Melnichenko, and L.~Molitor.
\newblock Selfish network creation with non-uniform edge cost.
\newblock {\em ArXiv}, 1706.10200, 2017.

\bibitem{corboparkes}
J.~Corbo and D.~C. Parkes.
\newblock The price of selfish behavior in bilateral network formation.
\newblock In {\em Proceedings of the 24th Symposium on Principles of
  Distributed Computing (PODC)}, pages 99--107. ACM, 2005.

\bibitem{demaine}
E.~D. Demaine, M.~T. Hajiaghayi, H.~Mahini, and M.~Zadimoghaddam.
\newblock The price of anarchy in network creation games.
\newblock {\em ACM Transactions on Algorithms}, 8(2):13:1--13:13, 2012.

\bibitem{derks2}
J.~Derks, J.~Kuipers, M.~Tennekes, and F.~Thuijsman.
\newblock Local dynamics in network formation.
\newblock Technical report, Maastricht University, 2008.

\bibitem{derks}
J.~Derks, J.~Kuipers, M.~Tennekes, and F.~Thuijsman.
\newblock Existence of {N}ash networks in the one-way flow model of network
  formation.
\newblock {\em Modeling, Computation and Optimization}, 6:9, 2009.

\bibitem{dutta}
B.~Dutta and S.~Mutuswami.
\newblock Stable networks.
\newblock {\em Journal of Economic Theory}, 76(2):322--344, 1997.

\bibitem{fabrikant}
A.~Fabrikant, A.~Luthra, E.~Maneva, C.~H. Papadimitriou, and S.~Shenker.
\newblock On a network creation game.
\newblock In {\em Proceedings of the 22nd Symposium on Principles of
  Distributed Computing (PODC)}, pages 347--351. ACM, 2003.

\bibitem{galeotti}
A.~Galeotti.
\newblock One-way flow networks: the role of heterogeneity.
\newblock {\em Economic Theory}, 29(1):163--179, 2006.

\bibitem{halleretal}
H.~Haller, J.~Kamphorst, and S.~Sarangi.
\newblock ({N}on-)existence and scope of {N}ash networks.
\newblock {\em Economic Theory}, 31(3):597--604, 2007.

\bibitem{hallersarangi}
H.~Haller and S.~Sarangi.
\newblock {N}ash networks with heterogeneous links.
\newblock {\em Mathematical Social Sciences}, 50(2):181--201, 2005.

\bibitem{hoffman}
A.~J. Hoffman and R.~R. Singleton.
\newblock On {M}oore graphs with diameters 2 and 3.
\newblock {\em IBM Journal of Research and Development}, 4(5):497--504, 1960.

\bibitem{jackson}
M.~Jackson and A.~van~den Nouweland.
\newblock Strongly stable networks.
\newblock {\em Games and Economic Behavior}, 51(2):420--444, 2005.

\bibitem{kang1975}
A.~N.~C. Kang and D.~A. Ault.
\newblock Some properties of a centroid of a free tree.
\newblock {\em Information Processing Letters}, 4(1):18 -- 20, 1975.

\bibitem{lenzner3}
B.~Kawald and P.~Lenzner.
\newblock On dynamics in selfish network creation.
\newblock In {\em Proceedings of the 25th Symposium on Parallelism in
  Algorithms and Architectures (SPAA)}, pages 83--92. ACM, 2013.

\bibitem{poa1}
E.~Koutsoupias and C.~H. Papadimitriou.
\newblock Worst-case equilibria.
\newblock In C.~Meinel and S.~Tison, editors, {\em STACS 99: 16th Annual
  Symposium on Theoretical Aspects of Computer Science Trier (Proceedings)},
  volume 1653 of {\em Lecture Notes in Computer Science}, pages 404--413.
  Springer, 1999.

\bibitem{poa2}
E.~Koutsoupias and C.~H. Papadimitriou.
\newblock Worst-case equilibria.
\newblock {\em Computer Science Review}, 3(2):65--69, 2009.

\bibitem{lenzner}
P.~Lenzner.
\newblock On dynamics in basic network creation games.
\newblock In G.~Persiano, editor, {\em Algorithmic Game Theory: 4th
  International Symposium, SAGT 2011 (Proceedings)}, volume 6982 of {\em
  Lecture Notes in Computer Science}, pages 254--265. Springer, 2011.

\bibitem{lenzner2}
P.~Lenzner.
\newblock Greedy selfish network creation.
\newblock In P.~W. Goldberg, editor, {\em Internet and Network Economics: 8th
  International Workshop, WINE 2012 (Proceedings)}, volume 7695 of {\em Lecture
  Notes in Computer Science}, pages 142--155. Springer, 2012.

\bibitem{matusz2}
A.~Mamageishvili, M.~Mihal{\'{a}}k, and D.~M{\"{u}}ller.
\newblock Tree {N}ash equilibria in the network creation game.
\newblock {\em Internet Mathematics}, 11(4-5):472--486, 2015.

\bibitem{buisan}
B.~Messegu\'{e}.
\newblock The price of anarchy in network creation.
\newblock Master's thesis, Universitat Polit{\`{e}}cnica de Catalunya, 2014.

\bibitem{matusz3}
M.~Mihal{\'{a}}k and J.~C. Schlegel.
\newblock Asymmetric swap-equilibrium: A unifying equilibrium concept for
  network creation games.
\newblock In B.~Rovan~V. Sassone and P.~Widmayer, editors, {\em Mathematical
  Foundations of Computer Science 2012: 37th International Symposium, MFCS 2012
  (Proceedings)}, volume 7464 of {\em Lecture Notes in Computer Science}, pages
  693--704. Springer, 2012.

\bibitem{matusz}
M.~Mihal{\'{a}}k and J.~C. Schlegel.
\newblock The price of anarchy in network creation games is (mostly) constant.
\newblock {\em Theory of Computing Systems (TOCS)}, 53(1):53--72, 2013.

\end{thebibliography}

\appendix

\section{Short Proof of Non-Existence of Strong Equilibria for $n \geq 5$ and $\alpha \in (1,2)$}\label{apx:nonexistence}

\begin{lemma}\label{nie ma}
 Let $\alpha \in (1,2)$ and $n \ge 5$. If $s \in \mathcal{S}$ is a strong equilibrium then $G(s)$ does not contain $K_{3}$ as a subgraph.
\end{lemma}

\begin{proof}
Fix a strong equilibrium $s \in \mathcal{S}$. Suppose that there are $3$ players $i, j, k$ such that the subgraph of $G(s)$ induced by $\{i, j, k\}$ is $K_{3}$. 
Without loss of generality we can assume that $j \in s_{i}$, $k \in s_{j}$ (and $i \in s_{k}$ or $k \in s_{i}$). 
Observe that $c_{i}(s) \le c_{i}(s_{i} \setminus \{j\}, s_{-i})$.
Since $c^{b}_{i}(s) - c^{b}_{i}(s_{i} \setminus \{j\}, s_{-i}) = \alpha > 1$, we have $c^{d}_{i}(s) - c^{d}_{i}(s_{i} \setminus \{j\}, s_{-i}) \le -2$. Let $V_{ij}$ be the set of players $v$ such that every shortest path from $i$ to $v$ in $G(s)$ contains the edge $(i,j)$. Observe that 
$|V_{ij}| \ge 2$ and $s_{i} \cap V_{ij} = \{j\}$. Hence $V_{ij} \setminus s_{i}$ is non-empty. 
Let $u \in V_{ij}$ be a neighbor of $j$ in $G(s)$ (see Figure~\ref{fig:new}) and observe that $d_{G(s)}(i,u) \ge 2$.
Analogously, there is a vertex $w$ that is a neighbor of $k$ in $G(s)$ such that every shortest path from $j$ to $w$ contains 
the edge $(j, k)$ (see Figure~\ref{fig:new}). Observe that $d_{G(s)}(i, w) \ge 2$ and $d_{G(s)}(u, w) \ge 2$, 
since every shortest path from $j$ to $w$ contains $(j, k)$ (not $(j, i)$ or $(j, u)$).
The players $\{i, u, v\}$ form an independent set of size $3$ in $G(s)$ and hence $s$ is not a strong equilibirum. This is in contradiction with Lemma~\ref{stw1}.
\begin{figure}
  \centering
  \begin{minipage}{0.4\textwidth}
      \raggedleft
        \includegraphics[scale=0.5]{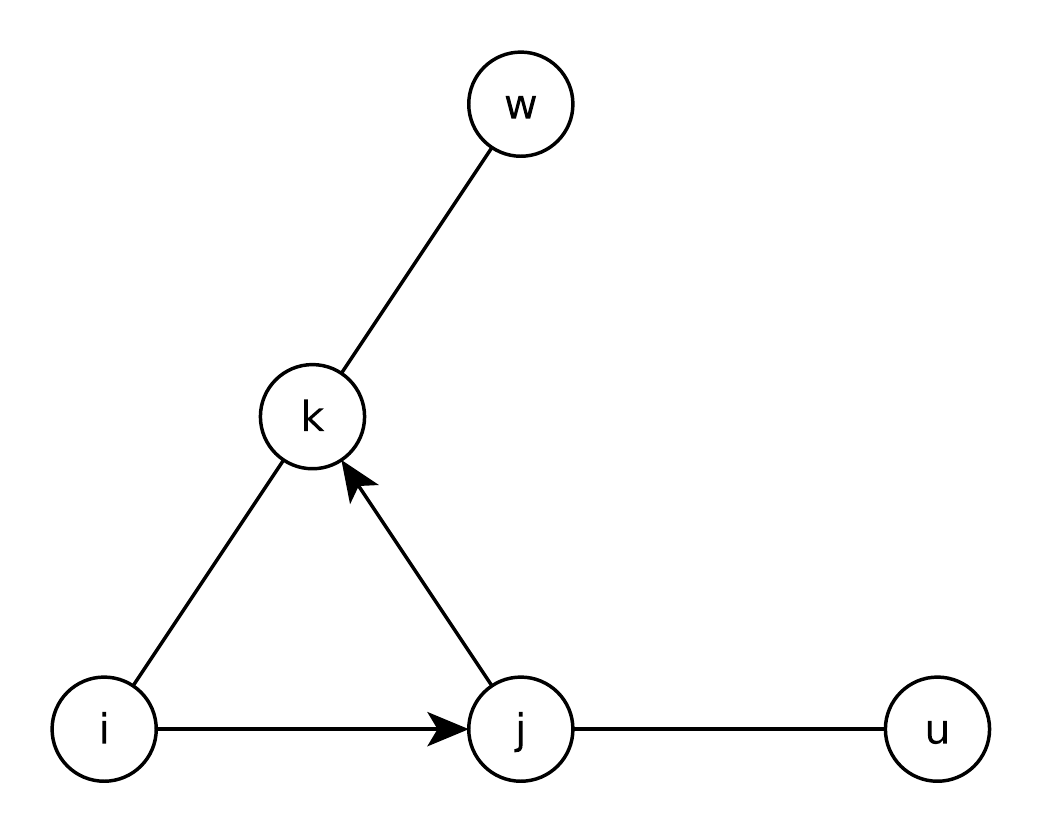}
        \caption{Depiction of the graph structure described in the proof of Lemma \ref{nie ma}.}\label{fig:new}
     \end{minipage}\hfill 
     \begin{minipage}{0.3\textwidth}
      \centering
        \scalebox{0.6}{\includegraphics{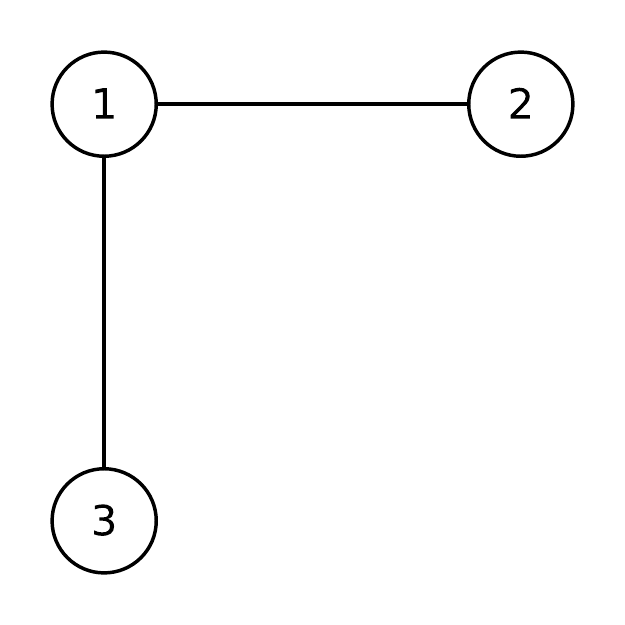}}
        \caption{A strong equilibrium, case $\alpha > 1$ and $n = 3$.}\label{rys2}
      \end{minipage}\hfill
\end{figure}
\end{proof}

\begin{theorem}\label{t508}
A network creation game where $\alpha \in (1,2)$ and $n \ge 5$ has no strong equilibrium.
\end{theorem}
\begin{proof}
Suppose $s \in \mathcal{S}$ is a strong equilibrium. From Lemma~\ref{stw1} the complement of $G(s)$, denoted by $\overline{G(s)}$, is a forest. By Lemma~\ref{nie ma} there is no independent set of cardinality $3$ in $\overline{G(s)}$.
  If $n \ge 5$, then there is no such forest, since the chromatic number of every forest is at most $2$ and it has an independent set of size $\lceil n/2 \rceil \geq 3$.
\end{proof}

\section{Strong Equilibria for $\alpha \in (1,2)$ and $n \in \{3,4\}$}\label{apx:structure}

\begin{lemma}\label{509}
 For $\alpha > 1$ and $n = 3$, a strategy profile  $s \in \mathcal{S}$ is a strong equilibrium if and only if $s$ is rational and $G(s)$ is a $3$-star (see Figure~\ref{rys2}).
\end{lemma}
\begin{proof}
Suppose $s \in \mathcal{S}$ is a Nash equilibrium. If $n = 3$, then, this Nash equilibrium forms a tree (as it is easy to see that it does not form the complete graph), hence it is a $3$-star. Suppose there is a coalition $K$ that can deviate to $s' = (s_{K}', s_{-K})$ such that $c_i(s') \leq c_i(s)$ for all $i \in K$. Suppose that $|s_{i}'| > |s_{i}|$ for some $i \in K$. Since $c_{i}^{b}(s') - c_{i}^{b}(s) \ge \alpha$ and $c_{i}(s') - c_{i}(s) < 0$, we have $c_{i}^{d}(s') - c_{i}^{d}(s) \le - 2$ which is impossible. Therefore, $|s_{i}'| \le |s_{i}|$ for all $i \in K$. Moreover, $|s_{i}'| = |s_{i}|$ since otherwise 
  $G(s')$ is disconnected. Since $G(s')$ is a tree consisting of $3$ vertices, any coalition of size at least $2$ contains
  a player $i$ such that $i$ is a leaf of $G(s')$. Since $|s_{i}'| = |s_{i}|$ and $c_{i}^{d}(s') = 3 \ge c_{i}^{d}(s)$, we have $c_{i}(s') \ge c_{i}(s)$, which means that not all players in $K$ decrease their cost when deviating to $s'$.
\end{proof}

\begin{lemma}\label{510}
For $\alpha \in (1,2)$ and $n = 4$, a strategy profile $s \in \mathcal{S}$ is a strong equilibrium if and only if $s$ is rational, $G(s)$ is a cycle, and $|s_{i}| = 1 = f(s, i)$ for all $i \in [n]$ (see Figure~\ref{rys3}).
\end{lemma}

\begin{proof}
Fix a rational strategy profile $s \in \mathcal{S}$ such that $G(s)$ is a cycle and $|s_{i}| = 1 = f(s, i)$ for all $i \in [n]$. Suppose there is a coalition $K$ that deviates to $s' = (s_{K}', s_{-K})$ for some $s_K'\in \mathcal{S}_K$. We prove that there is a player $i \in K$ such that $c_i(s') \geq c_i(s)$. Because $n = 4$ it holds that for all $s \in \mathcal{S}$, for all $i \in [n]$ we have $c_{i}^{d}(s) \ge 3$. Hence, there is no player $i$ such that $|s_{i}'| \ge 2$ because then $c_{i}(s') \ge 3 + 2 \alpha > 4 + \alpha = c_{i}(s)$. So $G(s')$ has at most $4$ edges. On the other hand $G(s')$ has at least $3$ edges (as otherwise it is disconnected). If $G(s')$ consists of $3$ edges, then it is either a linear graph or a $3$-star. In the first case observe that there is a player $i \in K$ such that $\text{deg}_{G(s')}(i) = 1$, so $c_{i}(s') \ge 6 > c_{i}(s)$. In the second case observe that (since $|s_{i}| \le 1$) there is a player $i \in K$ such that $|s_{i}'| = 1$ and $\deg_{G(s')}(i) = 1$, hence $c_{i}(s') \ge 5 + \alpha > c_{i}(s)$. If $G(s')$ consists of $4$ edges, then it is either a cycle $C_{4}$ or $K_{3}$ with the last vertex adjacent to one of the vertices forming subgraph $K_{3}$. In the first case this graph satisfies the assumption of the proposition because $|s_{i}| = 1$ for all $i \in [n]$. In the second case observe that there is a player $i \in K$ such that $|s'_{i}| = 1$ and $c_{i}(s') \ge 4 + \alpha \ge c_{i}(s)$. This proves that $s$ is a strong equilibrium.

It remains to prove that there are no other strong equilibria.
Without loss of generality, we can assume that the graph $G(s)$ is connected.  
The improving coalitions for any connected graph on $4$ vertices are presented in Figure~\ref{wsio}. 
\end{proof}

\begin{figure}
  \centering
    \begin{minipage}{0.4\textwidth}
      \centering
        \scalebox{0.6}{\includegraphics{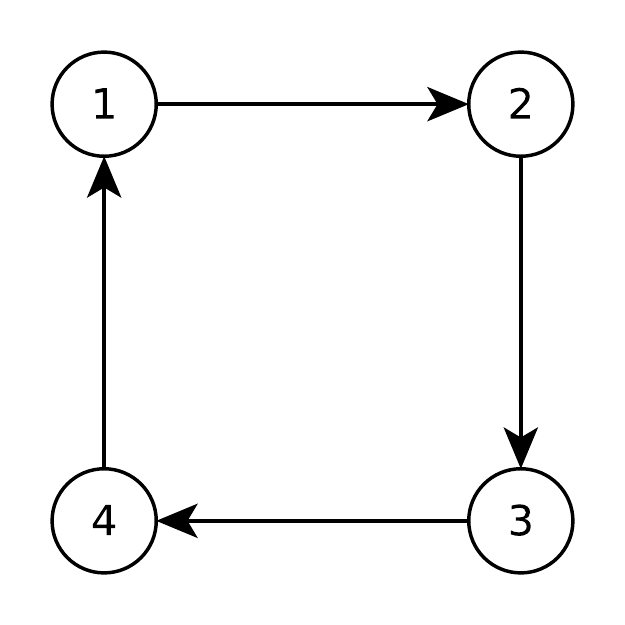}}
     \end{minipage}\hfill 
     \caption{A strong equilibrium, case $\alpha \in (1,2)$ and $n = 4$.}\label{rys3}
\end{figure}

\section{The Hoffman--Singleton Graph as a Pure Nash Equilibrium}
\label{hoffmansingleton}
We say that a pure Nash equilibrium $s$ is \emph{strict} if for each player $i$ and for all $s_i \in \mathcal{S}$, it holds that $c_i(s_i, s_{-i}) < c_i(s)$.
It is an interesting question to find the smallest strict Nash equilibrium which forms a non-tree graph for $\alpha > 1$.  

In \cite{albers}, the authors construct an example of a strict Nash equilibrium that is not a tree. The smallest game in their construction 
requires $210$ players. We will show here that there exists a non-tree strict equilibrium of $50$ vertices. 

\begin{definition}
The \emph{Hoffman--Singleton} graph can be defined as the unique $7$-regular graph with $50$ vertices such that every pair of adjacent vertices has no common neighbors and every pair of non-adjacent vertices has $1$ common neighbor. For its construction, we refer to \cite{hoffman}.
\end{definition}

\begin{theorem}
Let $\alpha \in (1, 26/3)$ and $n = 50$, then the network creation game has a strict Nash equilibrium $s \in \mathcal{S}$ such that $G(s)$ is not a tree. Specifically, $G(s)$ is the Hoffman--Singleton graph. 
\end{theorem}

\begin{proof}
By the construction of the Hoffman--Singleton graph, it is easy to see that there exist $s \in \mathcal{S}$ such that $G(s)$ is the Hoffman--Singleton graph and $|s_{i}| \in \{3, 4\}$ for all $i \in [n]$. Since the Hoffman--Singleton graph is symmetric, the total cost of a player depends only on the amount of edges she builds. Let $i \in [n]$ be a player such that $|s_{i}| = 4$. Suppose there is a strategy $s_{i}'$ such that $c_{i}(s_{i}',s_{-i}) < c_{i}(s)$ and $|s_{i}'| \geq 5$. Since the Hoffman--Singleton graph has diameter $2$, from~(\ref{eq:basic}) we have 
\begin{equation*}
c_{i}(s_{i}', s_{-i}) \geq 2 \cdot 50 - 2 - 3 + 5(\alpha - 1) = 90 + 5\alpha > 91 + 4 \alpha = c_{i}(s).
\end{equation*}
Hence there is no such strategy $s_{i}'$. Analogously, if $|s_{i}| = 3$ there is no improving strategy $s_{i}'$ such that $|s_{i}'| \geq 4$. Therefore it is enough to consider the strategies $s_{i}'$ of cardinality at most $4$. Enumeration of all of these strategies by hand (or by computer) then shows that for $\alpha \in (1, 26/3)$ a player cannot weakly improve her payoff, hence $s$ is a strict Nash equilibrium.
\end{proof}
One may wonder where the upper bound of $26/3$ on $\alpha$ comes from. The answer is as follows. For a player $i$ such that $|s_{i}| = 4$, there is a strategy $s_{i}'$ such that $|s_{i}'| = 1$ and $c_{i}^{d}(s_{i}', s_{-i}) = 117$. Since $c_{i}^{d}(s) = 91$, the player can exchange $3\alpha$ for a distance cost increase by $26$. 

\section{Strict Strong Equilibria}\label{apx:strictstrong}
We point out that a stronger solution concept of \emph{strict} strong equilibrium has received some attention in the literature, see e.g. \cite{dutta,jackson,Avrachenkov2016,Avrachenkov2016b}. Strict strong equilibria are strategy profiles for which no subset of players can deviate under a more permissive condition where at least one of the deviating players' costs strictly decreases, while all other costs of the deviating players do not increase. We briefly comment on the applicability of our structural results to strict strong equilibria.

For strict strong equilibria, the proof of Theorem \ref{se1} does not seem to immediately go through, while Theorem \ref{t508} obviously holds. Proposition \ref{prop:se12} remains to hold for the cases  for $n=3$ and $n=5$ (obviously for the latter case), Proposition \ref{prop:se12} does not hold for the $n=4$ case, and it is straightforward to show that there are no strict strong Nash equilibria for $n=4$ and $\alpha = (1,2)$: namely, the only strict strong equilibria can be ``directed 4-cycles'' as per our theorem for non-strict strong equilibria. A deviation that is allowed under strict strong equilibria is when two players on opposite sides of the cycle now deviate such that they form the middle vertices in a path of length 4. Theorem \ref{strongge2}, which showed that all stars are strong equilbria seems to hold for strict strong equilibria as well, but proving it would require additional arguments. We do not currently know whether Theorem \ref{thm:selowerbound} holds for strict strong equilibria, i.e., we do not know whether the given family of examples (Example 1) are strict strong equilibria, and proving that would then require additional arguments. We do conjecture that there exist non-star strict strong equilibria.

\end{document}